%% file: main.tex
\documentclass[aps,physrev,twocolumn,superscriptaddress]{revtex4-2}
\input{pretex}

\usepackage{amsmath,amsfonts,amssymb,array,graphicx,mathtools,multirow,bm,times,tcolorbox,relsize,booktabs, quantikz}
\usepackage[utf8]{inputenc}
\usepackage[T1]{fontenc}
\usepackage[paperwidth=199.8mm,paperheight=297mm,centering,hmargin=20mm,vmargin=2cm]{geometry}
\usepackage[ruled, vlined, linesnumbered]{algorithm2e}

\definecolor{colortwo}{rgb}{0.4,0.77,0.17}
\definecolor{colorthree}{rgb}{0.01,0.51,0.93}

\allowdisplaybreaks

\begin{document}
\title{A Nonstabilizerness Resource Law for Universal Quantum State Purification}

\author{Keming He}
\affiliation{QudeLeap Research, Shanghai 200030, China.}
\affiliation{Thrust of Artificial Intelligence, Information Hub, \\ The Hong Kong University of Science and Technology (Guangzhou), Guangdong 511453, China.}
\author{Enji Xiong}
\affiliation{QudeLeap Research, Shanghai 200030, China.}
\affiliation{Thrust of Artificial Intelligence, Information Hub, \\ The Hong Kong University of Science and Technology (Guangzhou), Guangdong 511453, China.}
\author{Xin Wang}
\email{felixxinwang@hkust-gz.edu.cn}
\affiliation{Thrust of Artificial Intelligence, Information Hub, \\ The Hong Kong University of Science and Technology (Guangzhou), Guangdong 511453, China.}

\date{\today}

\begin{abstract}

Quantum state purification aims to recover higher-fidelity quantum states from
multiple noisy copies and is a fundamental primitive for quantum information
processing. Magic resources enable operations beyond classically simulable
dynamics and are central to universal fault-tolerant quantum computation.
Recent no-go results show that classically simulable operations cannot achieve
a nontrivial universal fidelity gain. This motivates a quantitative theory of
the magic required for purification at prescribed success probability and
target fidelity. For universal purification with two input copies, we prove an exact linear
mana law in odd dimensions and a two-sided linear robustness law for
multi-qubit systems, which becomes exact for a single qubit. We also identify
an explicit successful purification map that makes the tradeoff transparent.
These results establish universal purification as a task obeying a quantitative
magic--fidelity law and link magic resources to error mitigation and
fault-tolerant quantum information processing.

\end{abstract}

\maketitle

\textit{Introduction}.---
Quantum devices are intrinsically noisy: as a state is prepared, processed, or stored, imperfections gradually drive it away from its intended trajectory. Quantum state purification offers a general remedy for this drift~\cite{Barenco1997,cirac1999optimal,keyl2001rate,fiuravsek2004optimal,fu2016quantum,yao2024protocols,li2024optimal,li2026quantum}. Given several noisy replicas of an unknown target state, the goal is to distill a output with higher fidelity to that target. In the universal setting, a protocol must work for arbitrary pure states, making purification a state-independent primitive for noise reduction rather than a scheme tailored to a particular input.

A simple but fundamental route to universal purification is symmetrization. Ideal replicas of the same pure state lie entirely within the symmetric subspace, whereas noise introduces components that leak outside it; repeatedly projecting the replicas back onto this subspace therefore suppresses errors without requiring any knowledge of the target state. Such symmetric-projection protocols—implementable through swap tests or controlled permutations—provide fundamental benchmarks for the achievable fidelity and success probability, and a substantial body of work has characterized their optimal performance under specific noise models and operational criteria. Unlike quantum error correction, which relies on carefully engineered codes and repeated syndrome extraction, symmetrization-based purification is hardware- and algorithm-agnostic, offering a complementary safeguard against noise accumulation in quantum computation.

The operational value of quantum state purification, however, depends on which transformations are physically accessible without incurring additional resource cost. In entanglement purification, for instance, LOCC and related free-operation classes determine the achievable distillation performance and reveal the resource character of the task~\cite{bennett1996purification,bennett1996mixed,deutsch1996quantum,dur2007entanglement,horodecki2009quantum}. An analogous question arises for general state purification~\cite{zhao2026power}: even when unrestricted quantum operations can improve a state, it remains unclear whether the same improvement can be achieved within a prescribed free class, or what additional resource it costs beyond that class. For probabilistic purification, this naturally leads to a quantitative question—posed at fixed success probability and target fidelity—of how much resource is truly required to purify a state.

A natural and key resource to examine in this context is magic, which lies at the heart of universal fault-tolerant quantum computation. Stabilizer operations, including Clifford gates, Pauli measurements, and stabilizer state preparations, are efficiently classically simulable by the Gottesman-Knill theorem~\cite{gottesman1997stabilizer}, while universal computation requires magic states or non-Clifford operations~\cite{gottesman1999demonstrating,zhou2000methodology,bravyi2005universal,Watson2016,Qassim2021}. Recent experiments have demonstrated key primitives for logical magic state processing, including distillation, cultivation, and code switching in fault-tolerant architectures~\cite{gupta2024encoding, sales2025experimental,rosenfeld2025magic, bluvstein2026fault}. This progress underscores the need for quantitative measures of magic, a need addressed by the resource theory of magic~\cite{wang2019quantifying,seddon2019quantifying,howard2017application,veitch2014resource,chen2025physical,seddon2021quantifying}. In the purification setting, recent impossibility results show that classically simulable operations cannot achieve a nontrivial universal fidelity gain~\cite{he2026no}. These results identify the zero-magic boundary of the task, but leave open the resource law beyond this boundary: how much magic is the price of a prescribed universal fidelity gain?


In this work, we answer this question by introducing magic resource quantifiers for successful purification under depolarizing noise. We assign a resource measure to each successful trace-non-increasing operation and show that the corresponding optimization admits a semidefinite programming (SDP) formulation. For purification with two input copies, we establish a nonstabilizerness resource law: in odd dimensions, the exponentiated mana is exactly linear in the fidelity gain over the unpurified state; in multi-qubit systems, the robustness of the Choi state is bounded above and below by analytic linear functions of the same fidelity gain, with equality attained in the single-qubit case. We further identify an explicit two-copy purification map achieving any prescribed success probability and target fidelity, whose parametrization renders the underlying tradeoff transparent. These results establish that universal purification obeys a quantitative magic–fidelity law~(cf.~Fig.~\ref{fig:framework}), linking purification performance to nonstabilizerness in a manner relevant to error mitigation and fault-tolerant quantum computation.


\begin{figure}[htbp]
    \centering
    \includegraphics[width=1.05\linewidth]{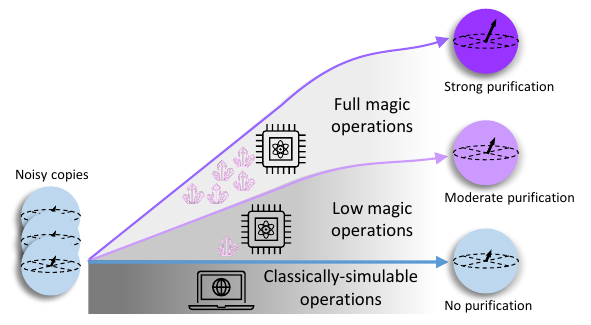}
    \caption{Schematic illustration of the role of magic in universal quantum state purification. Starting from noisy copies, classically simulable operations remain at the zero-gain boundary, operations with limited magic yield moderate purification, and operations with larger magic can reach stronger fidelity improvement.}
    \label{fig:framework}
\end{figure}


\textit{Preliminaries}.---
Let $\cH_A$ be a $d$-dimensional Hilbert space with computational basis $\{\ket{j}\}_{j=0}^{d-1}$. A quantum channel $\cN_{A\to B}$ is a completely positive and trace-preserving~(CPTP) map from $\cL(\cH_A)$ to $\cL(\cH_B)$. We also use completely positive trace-non-increasing~(CPTN) maps to describe successful branches of probabilistic protocols. The depolarizing channel with error parameter \(\delta\) is denoted by \(\cD_\delta(\rho)=(1-\delta)\rho+\delta I_d/d\). The Choi-Jamio\l{}kowski operator of \(\cN_{A\to B}\) is \(J^\cN_{AB}\coloneqq\sum_{i,j=0}^{d-1}\ketbra{i}{j}_A\otimes\cN_{A\to B}(\ketbra{i}{j})\), and the normalized Choi state is \(\Phi^\cN_{AB}\coloneqq J^\cN_{AB}/d_A\).

Magic, or nonstabilizerness, is the resource that separates general quantum states and operations from those generated by stabilizer states and Clifford operations. In odd-dimensional systems, this resource is naturally described by the discrete Wigner representation~\cite{gross2006hudson,veitch2014resource,wang2019quantifying}. For prime \(d\), the Heisenberg--Weyl operators are \(T_{\mathbf u}=\tau^{-a_1a_2}Z^{a_1}X^{a_2}\), where \(\tau=e^{(d+1)\pi i/d}\), \(\mathbf u=(a_1,a_2)\in\mathbb Z_d\times\mathbb Z_d\), and \(X,Z\) are the shift and boost operators. The phase space point operators are defined by \(A^{\mathbf 0}=d^{-1}\sum_{\mathbf u}T_{\mathbf u}\) and \(A^{\mathbf u}=T_{\mathbf u}A^{\mathbf 0}T_{\mathbf u}^\dagger\). For composite odd dimensions, we use the corresponding tensor-product construction. The Wigner function of a state is \(W_\rho(\mathbf u)=d^{-1}\tr[A^{\mathbf u}\rho]\). By the discrete Hudson theorem, a pure state in odd dimension is a stabilizer state if and only if its Wigner function is non-negative~\cite{gross2006hudson}. Therefore, Wigner negativity gives a direct witness of magic.

For operations, we use the channel Wigner function \(W_\cN(\mathbf v|\mathbf u)\), defined from the Choi operator of \(\cN\), and quantify operational magic by the mana \(\cM(\cN)=\log\max_{\mathbf u}\sum_{\mathbf v}|W_\cN(\mathbf v|\mathbf u)|\). This definition extends directly to trace-non-increasing maps, which are the objects that appear after conditioning on success in probabilistic purification. For multi-qubit systems, the discrete Wigner representation is not equally convenient as a magic witness while keeping the usual stabilizer operations free~\cite{delfosse2015wigner,raussendorf2017contextuality,raussendorf2020phase}. We therefore use the robustness of magic \(\cR(\rho)\), which measures the minimal stabilizer decomposition cost of a state \(\rho\)~\cite{howard2017application,seddon2019quantifying}. For a map \(\cN\), we quantify its Choi-state magic by \(\cR(\Phi^\cN_{AB})\). In the probabilistic setting, the relevant object is the Choi-state robustness of the successful trace-non-increasing operation. Further details are given in Appendix~\ref{appendix:preliminaries}.

\textit{Magic cost of probabilistic purification}.---
Noise and decoherence degrade quantum states prepared on near-term and fault-tolerant devices. Quantum state purification addresses this problem by using several noisy copies of an unknown target state to produce an output with higher fidelity. This task is usually formulated as a performance optimization problem, asking which fidelity and success probability can be achieved by a given class of protocols~\cite{fiuravsek2004optimal}. In the universal setting, however, this performance question is constrained by resource theory: classically simulable operations cannot achieve a nontrivial universal fidelity gain~\cite{he2026no}. This motivates the quantitative question studied here. Once a target success probability and fidelity are prescribed, how much magic must be present in the successful branch that realizes them? The same optimization framework can be
formulated for more general noise models once the noisy input ensemble is
specified, but in this paper we focus on depolarizing noise, for which the
universal resource laws can be made explicit.


Let \(\ket{\psi}\) be the ideal pure state and write \(\psi=\ketbra{\psi}{\psi}\). We assume that each copy undergoes the depolarizing channel \(\cD_\delta\), so the fidelity of one noisy copy with the target is \(\lambda_0\coloneqq\tr[\cD_\delta(\psi)\psi]=1-\frac{d-1}{d}\delta\). Given \(n\) noisy copies, a successful branch of a probabilistic protocol is described by a CPTN map \(\cE_{A^n\to A}\), with unnormalized output \(\sigma_\psi=\cE_{A^n\to A}(\cD_\delta(\psi)^{\ox n})\).

For a finite test set \(\Psi=\{\psi_i\}_{i=1}^{|\Psi|}\), we define the average conditional fidelity and average success probability~\cite{fiuravsek2004optimal, yao2024protocols} as
\begin{equation}\label{eq:max_fidelity}
\begin{aligned}
    F(\cE;\Psi)
    &=
    \frac{\sum_{\psi_i\in\Psi}\tr[\sigma_{\psi_i}\psi_i]}
    {\sum_{\psi_i\in\Psi}\tr[\sigma_{\psi_i}]},
    \\
    P(\cE;\Psi)
    &=
    \frac{1}{|\Psi|}
    \sum_{\psi_i\in\Psi}\tr[\sigma_{\psi_i}].
\end{aligned}
\end{equation}
For a prescribed target pair \((p,f)\), we require \(F(\cE;\Psi)=f\) and \(P(\cE;\Psi)=p\). We focus on the purification regime \(f\ge\lambda_0\), where the conditional output improves on the original noisy copy. The corresponding feasible set is
\begin{equation}
\begin{aligned}
\mathcal A_\delta(n,f,p;\Psi)
\coloneqq
\Big\{
\cE_{A^n\to A}:\;
\cE \ {\rm is\ CPTN},\\
F(\cE;\Psi)=f,\;
P(\cE;\Psi)=p
\Big\}.
\end{aligned}
\end{equation}

Fixing \((p,f)\) changes the role of purification from finding the best attainable performance to pricing a desired operational target. A less demanding target may require less nonstabilizerness in the successful branch, while a larger fidelity gain should require more. We make this idea precise by minimizing a magic measure over all operations in \(\mathcal A_\delta(n,f,p;\Psi)\).

For odd-dimensional qudit systems, we define the \textit{mana of purification for $\Psi$} at target pair \((p,f)\) by
\begin{equation}
\begin{aligned}
\cM_{\cD_\delta}(n,f,p;\Psi)
&\coloneqq
-\log p
+
\min_{\cE\in\mathcal A_\delta(n,f,p;\Psi)}
\cM(\cE) \\
=
\log &
\min_{\cE\in\mathcal A_\delta(n,f,p;\Psi)}
\max_{\mathbf u}
\frac{1}{p}
\sum_{\mathbf v}
\left|W_{\cE}(\mathbf v|\mathbf u)\right|.
\end{aligned}
\end{equation}
For multi-qubit systems, we define the \textit{robustness of purification for $\Psi$} at target pair \((p,f)\) by
\begin{equation}
    \cR_{\cD_\delta}(n,f,p;\Psi)
    \coloneqq
    \min_{\cE\in\mathcal A_\delta(n,f,p;\Psi)}
    \frac{1}{p}
    \cR(\Phi^\cE_{A_I^nA_O}),
\end{equation}
where \(\Phi^\cE_{A_I^nA_O}=J^\cE_{A_I^nA_O}/d_{A_I^n}\) is the Choi state of the successful branch. The factor \(1/p\) converts the unnormalized successful CPTN branch into the
map conditioned on success. This is the natural quantity for a probabilistic
protocol, because it separates the successful purification map itself from its
success probability
~\cite{vidal2000entanglement,regula2022probabilistic,
fang2018probabilistic,takagi2019general}. \(\cM_{\cD_\delta}\) and
\(\cR_{\cD_\delta}\) quantify the magic of the successful purification map
itself, while \(p\) records its success probability. This separation between
branch magic and success probability becomes quantitative in the universal
two-copy depolarizing setting studied below.

Postselection can trade rate for fidelity, but it does not remove the resource cost of a fidelity gain. The definitions above therefore give a resource analogue of purification: rather than asking only whether \((p,f)\) is achievable, we ask for the minimum nonstabilizerness needed to implement it.


We extend the finite-set task to universal purification by replacing the finite average over \(\Psi\) with the Haar average over all pure input states. This gives the \textit{mana of universal purification} \(\cM_{\cD_\delta}(n,f,p,d)\) and the \textit{robustness of universal purification} \(\cR_{\cD_\delta}(n,f,p,d)\). After expressing the fidelity and success constraints as linear constraints on the Choi operator of \(\cE_{A^n\to A}\), both versions admit semidefinite programming formulations. The detailed definitions and SDP derivations are given in Appendix~\ref{app:sdp}.

\textit{Resource laws for universal purification}.---
We now specialize to two input copies, \(n=2\), and study the mana and robustness of universal purification. By definition, \(\cM_{\cD_\delta}(2,f,p,d)\) and
\(\cR_{\cD_\delta}(2,f,p,d)\) are obtained by optimizing over all feasible
successful trace-non-increasing branches satisfying the target pair
\((p,f)\). The theorems below therefore give global resource laws. The main result is that the magic required by the successful branch is controlled directly by the universal fidelity gain \(f-\lambda_0\). 

For odd-dimensional qudit systems, the exponentiated mana admits an exact linear law.

\begin{theorem}\label{thm:mana-universal}{\rm (Mana law for universal purification)}
For two-to-one universal probabilistic purification of odd-dimensional qudit states under depolarizing noise, the exponentiated mana of universal purification is exactly linear in the fidelity gain \(f-\lambda_0\).
\end{theorem}


For multi-qubit systems, the corresponding resource law is expressed in terms of Choi-state robustness.


\begin{theorem}\label{thm:robustness-universal}{\rm (Robustness law for universal purification)}
For two-to-one universal probabilistic purification of multi-qubit states under depolarizing noise, the robustness of universal purification is bounded above and below by linear functions of the fidelity gain \(f-\lambda_0\). For a single-qubit system, the two bounds coincide, giving an exact linear law.
\end{theorem}

For odd-dimensional systems, the exact slope is given explicitly in
Appendix~\ref{app:mana-proof}; for multi-qubit systems, the lower and upper
slopes are given in Appendix~\ref{app:robustness-proof}.

The robustness law places our previous single-qubit no-go result~\cite{he2026no}
into a broader quantitative picture. At \(f=\lambda_0\), the lower and upper robustness bounds both reduce to one, so \(\cR_{\cD_\delta}(2,\lambda_0,p,2^m)=1\), where $m\ge 1$ is the number of qubits. For any nontrivial purification target with \(f>\lambda_0\), the lower bound becomes strictly larger than one. Thus completely stabilizer preserving operations are confined to the zero-gain point and cannot realize two-to-one universal probabilistic purification with a fidelity gain.

\begin{corollary}\label{cor:multiqubit-no-go-quantitative}
There is no two-to-one universal probabilistic purification protocol using completely stabilizer preserving operations for multi-qubit states under depolarizing noise.
\end{corollary}

Together, Theorems~\ref{thm:mana-universal} and~\ref{thm:robustness-universal} convert the qualitative no-go statement for universal purification by classically simulable operations~\cite{he2026no} into a quantitative magic--fidelity law. The zero-resource point is fixed at the fidelity
\(\lambda_0\) of the unpurified state. Any universal improvement beyond this point requires positive magic in the operation conditioned on success, and the amount of required magic is controlled by the fidelity gain \(f-\lambda_0\). Thus the theorems give a fundamental resource limit for two-copy universal purification under depolarizing noise. 


We sketch the proofs of Theorems~\ref{thm:mana-universal} and~\ref{thm:robustness-universal} as follows. For the mana law, the lower bound is obtained from the dual SDP after using the symmetry of the universal task. The matching upper bound is attained by the explicit two-copy branch introduced below in Eq.~\eqref{eq:fixed-pf-projection-map}, with coefficients chosen to realize the prescribed success probability and fidelity. Since the two bounds coincide, Theorem~\ref{thm:mana-universal} gives the exact odd-dimensional mana law.

For the robustness law, the upper bound uses the same two-copy branch together with a stabilizer decomposition of its Choi state. The lower bound is obtained by Clifford twirling an arbitrary feasible map without changing the purification constraints, evaluating a robustness witness on the resulting invariant Choi state, and using the fact that any tripartite qubit stabilizer state is locally Clifford equivalent to a tensor product of GHZ states, Bell pairs, and single-qubit stabilizer states~\cite{bravyi2006ghz,looi2011tripartite}. This structure lets us bound the relevant stabilizer contributions explicitly. The complete proofs are given in Appendices~\ref{app:mana-proof} and~\ref{app:robustness-proof}.

Figure~\ref{fig:resource-law} illustrates these resource laws for \(\delta = 0.5\).
The laws quantify the operation conditioned on success. 
The factor \(1/p\) converts the unnormalized successful operation into the resource content of the selected branch. 
Thus \(p\) enters through the feasibility of the target pair \((p,f)\) and determines the rate at which the branch is obtained, while the mana or Choi-state robustness of the successful branch is controlled by the fidelity gain \(f-\lambda_0\).

\begin{figure}[t]
\centering
\includegraphics[width=1.02\columnwidth]{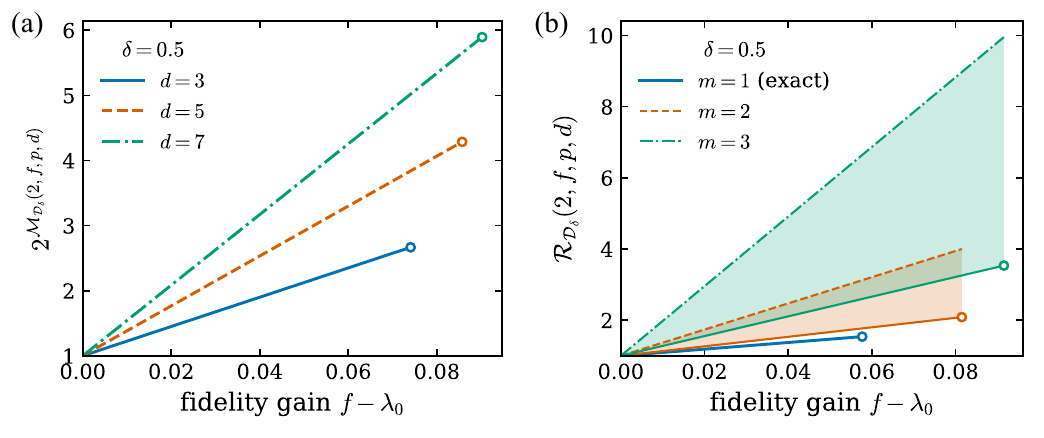}
\caption{Resource laws for two-copy universal purification at \(\delta=0.5\). The exponentiated mana and the robustness bounds grow linearly with the fidelity gain \(f-\lambda_0\). In panel (b), shaded regions indicate the intervals between lower and upper robustness bounds.}
\label{fig:resource-law}
\end{figure}

These proof structures also suggest a broader program. The two-copy
depolarizing setting is the first analytically tractable regime in which the
resource law can be made explicit. The same optimization framework extends to more input copies and more general noise models, but the relevant symmetry and positivity structures quickly become richer, making both direct numerical computation and analytic bounds substantially harder. On the
achievability side, one needs higher-copy purification branches together
with stabilizer-state decompositions of their Choi states, whose complexity
grows super-exponentially with the total number of qubits. Reducing this
complexity would require a better understanding of the relation among
multi-copy stabilizer states, permutation operators, and 
entanglement~\cite{nezami2020multipartite}. On the converse
side, one needs a sharper characterization of higher-copy Clifford
commutant algebras and their positivity. These algebras become
substantially richer at higher tensor powers~\cite{zhu2017multiqubit,
gross2021schur,bittel2025complete}, making their invariant structures and
positivity constraints central to extending the resource law beyond the
two-copy setting.

\textit{Resource tradeoff for the successful purification map}.---
We now make the successful two-copy purification map explicit and use it to
interpret the resource tradeoff. A
natural reference protocol is given by the swap test. Let
\(\mathbf P_2((12))\) be the swap operator between the two input systems. The
swap test measures the two projectors $\Pi_{\mathrm{sym}}
=
\left(I_2+\mathbf P_2((12))\right)/2,
\;
\Pi_{\mathrm{asym}}
=
\left(I_2-\mathbf P_2((12))\right)/2$,
which project onto the symmetric and antisymmetric subspaces. Accepting only
the symmetric outcome gives the usual symmetric projection protocol and reaches
the extremal high fidelity point. For a general feasible target pair \((p,f)\),
the operation conditioned on success need not be purely symmetric. An explicit
map that achieves the target \((p,f)\) is given by a weighted combination of
the two projection branches:
\begin{equation}\label{eq:fixed-pf-projection-map}
\begin{aligned}  
    \cE_{A_I^2\to A_O}(\cdot)
    =
    \tr_{A_{O, 2}}\big[
    2(\mu_1+\mu_2)
    \Pi_{\mathrm{sym}}(\cdot)\Pi_{\mathrm{sym}}^\dagger  \\
    + 
    2(\mu_1-\mu_2)
    \Pi_{\mathrm{asym}}(\cdot)\Pi_{\mathrm{asym}}^\dagger
    \big].
    \end{aligned}
\end{equation}
Here the coefficients \(\mu_1\) and \(\mu_2\) are fixed by the prescribed
success probability \(p\) and conditional fidelity \(f\), and
$\tr_{A_{O, 2}}$ denotes tracing out the second copy after the projection. Their
explicit expressions, together with the verification that
Eq.~\eqref{eq:fixed-pf-projection-map} satisfies the desired constraints, are
given in Appendix~\ref{app:mana-proof}.

This map explains the operational origin of the resource laws. Writing
\(t=\mu_2/\mu_1\) and \(s=2(\mu_1+\mu_2)\), Eq.~\eqref{eq: p f depending on s t}
in Appendix~\ref{app:mana-proof} shows that the conditional fidelity is fixed
by \(t\), whereas the success probability is rescaled by \(s\). Figure~\ref{fig:selected-branch-parameter-space} visualizes this separation in the \((t,s)\) plane: the background color, given by \(2^{\cM_{\cD_\delta}(2,f,p,3)}-1\), depends only on \(t\), whereas the contour lines encode the success probability \(p\), and the inset shows how increasing \(t\) suppresses the antisymmetric branch along \(s=1\). Increasing \(t\) suppresses the antisymmetric branch and raises the universal
fidelity gain, and the resource laws associate this change with a larger magic
cost for the selected operation. By contrast,
changing \(s\) changes only the probability of obtaining that selected
operation. At \(t=1\), the antisymmetric branch is fully suppressed and the branch reduces
to the symmetric projection protocol. This gives the golden point
\((p^\star,f^\star)\) of the two-copy purification tradeoff~\cite{yao2024protocols}. For a target
success probability \(p\le p^\star\), one can keep \(t=1\) and choose
\(s=p/p^\star\), so the conditional fidelity remains \(f^\star\). For
\(p>p^\star\), keeping \(t=1\) would require \(s>1\), which is forbidden by
trace non-increase. The feasible branch must then move away from the symmetric projection by taking
\(s=1\) and \(t<1\). The resulting antisymmetric contribution decreases the
universal fidelity gain, and the resource laws imply a correspondingly smaller
magic cost for the selected operation, exactly for the exponentiated mana in
odd dimensions and within the robustness bounds for multi-qubit systems.

\begin{figure}[t]
    \centering
    \includegraphics[width=0.8\linewidth]{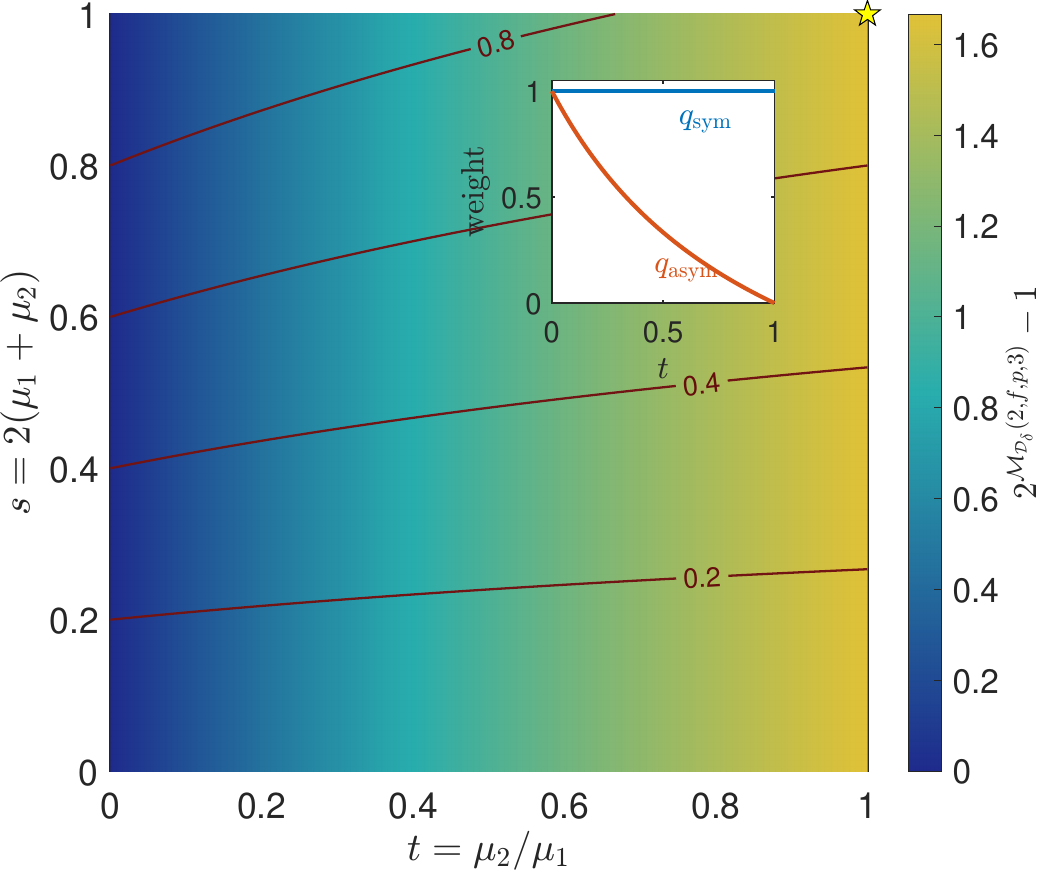}
    \caption{Parameter space structure of the successful purification map for odd-dimensional universal purification at \(\delta=0.5\) and \(d=3\). The background color shows \(2^{\cM_{\cD_\delta}(2,f,p,3)}-1\), and the contour lines show the success probability \(p\). The point \((t,s)=(1,1)\) is the preimage of the golden point. Inset: retained branch weights at \(s=1\), where \(q_{\mathrm{sym}}\coloneqq 2(\mu_1+\mu_2)=s\) and \(q_{\mathrm{asym}}\coloneqq 2(\mu_1-\mu_2)=s(1-t)/(1+t)\). Increasing \(t\) suppresses the antisymmetric branch and increases both the fidelity gain and the magic of the selected operation.}
    \label{fig:selected-branch-parameter-space}
\end{figure}

This map can be realized by an ordinary swap test followed by an additional
acceptance step. The corresponding implementing unitary may contain
non-Clifford ingredients, depending on how the acceptance step is performed.
The resource measures studied here, however, are assigned to the
trace-non-increasing operation obtained after measuring the accept flag,
tracing out the unused systems, and conditioning on success, rather than to a
particular unitary implementation of the whole instrument. Thus lowering the
target fidelity can reduce the mana or Choi-state robustness of the selected
operation, even if a concrete implementation contains additional
non-Clifford structure.

This is the natural level for probabilistic resource theories. A selective
process has several possible outcomes, and an individual accepted outcome may
carry a different amount of resource from the other outcomes or from the full
instrument. Probabilistic resource theories therefore distinguish the resource
of the overall process, the average resource over all branches, and the
resource of a specified postselected branch
~\cite{vidal2000entanglement,regula2022probabilistic,
fang2018probabilistic,takagi2019general}. Related distinctions between
resources invested in a measurement process and resources associated with
measurement outcomes also appear in studies of measurement-induced
magic~\cite{li2025invested}. In the present task, the accepted branch is precisely the purification
operation whose magic enters the resource laws. The resource
laws show that whenever this accepted branch gives a universal fidelity gain,
the accepted operation carries a corresponding amount of magic.

\textit{Concluding remarks}.---
We investigated the magic cost of probabilistic quantum state purification with
prescribed success probability and fidelity. We formulated mana and
robustness measures for the successful trace-non-increasing operation and
derived SDP characterizations for them. For two-copy universal purification,
we proved an exact linear law for the exponentiated mana in odd-dimensional
systems and two-sided linear robustness bounds for multi-qubit systems. A key
technical insight behind the robustness law is the interplay between symmetry
reduction and stabilizer-state structure, which allows both upper and lower
bounds to be controlled analytically and yields an exact law in the
single-qubit case. These results show that any universal fidelity improvement
beyond the unpurified value requires magic in the successful operation, and
that the required amount is controlled directly by the fidelity gain. In this
sense, two-copy universal purification obeys a concrete resource--performance
law for a probabilistic quantum information task.

Our results also support the broader view that magic is relevant to a wider
range of operational tasks beyond universal quantum computation. Recent works have connected nonstabilizerness to quantum capacity~\cite{bu2025magic}, state discrimination under restricted operations~\cite{zhu2024limitations,kwon2025nonstabilizerness}, measurement-conditioned many-body dynamics and monitored circuits~\cite{li2025invested,zhen2026invariant}, quantum-enhanced metrology~\cite{hernandez2026nonstabilizerness}, and algorithmic complexity in Shor's factoring task~\cite{paviglianiti2026true}. In this context, our results identify universal purification as another basic probabilistic task whose achievable performance is quantitatively governed by magic. These developments suggest that understanding how nonstabilizerness governs
operational advantages across quantum information tasks will remain an
important direction for future study in the resource theory of magic.


\textit{Acknowledgements}.---
This work was partially supported by the National Natural Science Foundation of China (Grant Nos.~92576114, 12447107) and the Guangdong Provincial Quantum Science Strategic Initiative (Grant Nos.~GDZX2403008, GDZX2503001, and GDZX2403001). E.X. would like to thank the support from the X Program from HKUST(Guangzhou).

\bibliography{ref}

\onecolumngrid

\appendix
\setcounter{subsection}{0}
\setcounter{table}{0}
\setcounter{figure}{0}

\vspace{3cm}

\begin{center}
\Large{\textbf{Appendix --- A Nonstabilizerness Resource Law for Universal Quantum State Purification} 
}
\end{center}

\renewcommand{\theequation}{S\arabic{equation}}
\renewcommand{\theproposition}{S\arabic{proposition}}
\renewcommand{\thedefinition}{S\arabic{definition}}
\renewcommand{\thefigure}{S\arabic{figure}}
\setcounter{equation}{0}
\setcounter{table}{0}
\setcounter{section}{0}
\setcounter{proposition}{0}
\setcounter{definition}{0}
\setcounter{figure}{0}


In this Supplemental Material, we provide the detailed definitions and proofs supporting the results in the manuscript ``A Nonstabilizerness Resource Law for Universal Quantum State Purification.'' Appendix~\ref{appendix:preliminaries} collects the notation and background on the discrete Wigner representation, robustness of magic, and partially transposed permutation algebras used throughout the proofs. Appendix~\ref{app:sdp} formulates the mana and robustness of purification as semidefinite programs by expressing the fidelity and success constraints in terms of the Choi operator. Appendix~\ref{app:mana-proof} proves the explicit form of the exact linear mana law for two-to-one universal purification in odd-dimensional systems. Appendix~\ref{app:robustness-proof} proves the explicit two-sided linear robustness law for multi-qubit systems, including the exact single-qubit case.

\section{Preliminaries}\label{appendix:preliminaries}

We consider a finite-dimensional Hilbert space $\cH_A$ representing the quantum system $A$ with dimension $d$. Let $\{\ket{j} \}_{j=0,\cdots,d-1}$ be a standard computational basis. We use $\cL(\cH_A)$ to represent the set of linear operators that map from $\cH_A$ to itself. A density operator is a positive semidefinite operator in $\cL(\cH_A)$ with trace one, and $\cD(\cH_A)$ denotes the set of all density operators in $\cH_A$. A quantum channel $\cN_{A\to B}$ is a linear map from $\cL(\cH_A)$ to $\cL(\cH_{B})$ that is completely positive and trace-preserving (CPTP). We also introduce a quantum operation that is completely positive and trace non-increasing (CPTN). Their associated Choi-Jamio\l{}kowski operators are $J^{\cN}_{AB}\coloneqq \sum_{i, j=0}^{d-1}\ketbra{i}{j} \ox \cN_{A \to B}(\ketbra{i}{j})$~\cite{choi1975completely, jamiolkowski1972linear}. We denote $\cI$ as the identity channel, and $\cD_\delta$ as a depolarizing channel with error parameter $\delta$. The symmetric group of degree $n$ is denoted by $\cS_n$  and $\mathbf{P}_n(c)$ represents the permutation operator for $c\in\cS_n$. We denote $\Pi_n\coloneqq\frac{1}{n!}\sum_{c\in\cS_n}\mathbf{P}_n(c)$ as the projector on the symmetric subspace of $\cH_{A}^{\ox n}$.

\subsection{The discrete Wigner function}
We recall the discrete Wigner representation used in this work~\cite{gross2006hudson,veitch2014resource,wang2019quantifying}.
For prime dimension $d$, the unitary boost and shift operators $X,Z\in\cL(\cH)$ are defined by
\begin{equation}
    X\ket{j}=\ket{j\oplus1}, \quad Z\ket{j}=\omega^j\ket{j},
\end{equation}
where $\omega=e^{2\pi i/d}$ and $\oplus$ denotes addition modulo $d$. The Heisenberg--Weyl operators are defined as 
\begin{equation}
    T_{\mathbf{u}} = \tau^{-a_1a_2}Z^{a_1}X^{a_2},
\end{equation}
where $\tau=e^{(d+1)\pi i/d}$, $\mathbf{u}=(a_1, a_2)\in \ZZ_d \times \ZZ_d$. For the composite system $\cH_A \ox \cH_B$, the Heisenberg--Weyl operators are the tensor product of the subsystem Heisenberg--Weyl operators:
\begin{equation}
    T_{\mathbf{u}_A \oplus \mathbf{u}_B} = T_{\mathbf{u}_A} \ox T_{\mathbf{u}_B},
\end{equation}
where $\mathbf{u}_A \oplus \mathbf{u}_B \in \ZZ_d \times \ZZ_d \times \ZZ_d \times \ZZ_d$. For each point $\mathbf{u}$ in the discrete phase space, there is a corresponding phase space point operator $A^{\mathbf{u}}$ defined as 
\begin{equation}
    A^{\mathbf{0}} \coloneqq\frac{1}{d}\sum_\mathbf{u}T_{\mathbf{u}}, \quad A^{\mathbf{u}} \coloneqq T_{\mathbf{u}}A_{\mathbf{0}}T_{\mathbf{u}}^\dagger.
\end{equation}
The discrete Wigner function of a state $\rho$ at the point $\mathbf{u}$ is then defined as
\begin{equation}
    W_{\rho}(\mathbf{u})=\frac{1}{d}\tr [A^{\mathbf{u}}\rho].
\end{equation}
More generally, we can replace $\rho$ with a Hermitian operator $H$ for the discrete Wigner function. Some useful properties are listed:
\begin{enumerate}
    \item $A^{\mathbf{u}}$ is Hermitian;
    \item $\sum_\mathbf{u}A^\mathbf{u}/d=I$;
    \item $\tr[A^\mathbf{u}A^{\mathbf{u'}}]=d\delta(\mathbf{u}, \mathbf{u'})$, where $\delta(a,b)$ is the discrete Dirac delta function;
    \item $\tr[A^\mathbf{u}]=1$;
    \item $H = \sum_{\mathbf{u}}W_H(\mathbf{u})A^\mathbf{u}$;
    \item $\{A^\mathbf{u}\}_\mathbf{u} = \{(A^\mathbf{u})^T\}_\mathbf{u}$;
    \item For the composite system $\cH_A \ox \cH_B$, the phase space point operators are the tensor product of the subsystem phase space point operators $A^{\mathbf{u}_A\oplus \mathbf{u}_B} = A^{\mathbf{u}_A} \ox A^{\mathbf{u}_B}$.
\end{enumerate}

A Hermitian operator $H$ has non-negative discrete Wigner functions if $\forall \mathbf{u}, W_H(\mathbf{u})\geq 0$. For odd dimensions, according to the discrete Hudson's theorem~\cite{gross2006hudson}, a pure state is a stabilizer state if and only if it has non-negative discrete Wigner functions. 

The sum negativity and mana are two kinds of magic witness in terms of the discrete Wigner function~\cite{veitch2014resource}. The sum negativity of a state $\rho$ is defined as 
\begin{equation}
    \text{sn}(\rho) = \sum_{\mathbf{u}: W_\rho(\mathbf{u}) < 0}|W_\rho(\mathbf{u})| = \frac{1}{2} \left(\sum_{\mathbf{u}}|W_\rho(\mathbf{u})| -1  \right).
\end{equation}
Mana is defined as
\begin{equation}
    \cM(\rho) = \log \left( \sum_{\mathbf{u}}|W_\rho(\mathbf{u})| \right) = \log (2\cdot \text{sn}(\rho) + 1).
\end{equation}

The discrete Wigner function of a state can be naturally extended to quantum channels. We recall the definition of the discrete Wigner function of a quantum channel from~\cite{mari2012positive} (see the discussion surrounding~\cite[Eq.~(10)]{mari2012positive}), which is related to the Wigner function of a quantum channel as defined in~\cite[Eq.~(95)]{Bartlett_2012}, and formally defined in~\cite{wang2019quantifying}.

\begin{definition}{\rm (Discrete Wigner function of a quantum channel~\cite{wang2019quantifying})}
Given a quantum channel $\cN_{A\rightarrow B}$, its discrete Wigner function is defined as 
\begin{equation}
\begin{aligned}
    W_{\cN}(\mathbf{v}|\mathbf{u}) &\coloneqq \frac{1}{d_B}\tr \big[(A^\mathbf{u}_A \ox A^\mathbf{v}_B) J^\cN_{AB}\big] \\
    &= \frac{1}{d_B}\tr \big[A^\mathbf{v}_B\cN(A^\mathbf{u}_A)\big].
\end{aligned}
\end{equation}
\end{definition}

\begin{definition}
    [Mana of a quantum channel]\label{def: channel mana} The mana of a quantum channel $\cN$ is defined as
    \begin{equation}
    \cM(\cN) =  \log \max_\mathbf{u} \frac{1}{d_B} \sum_\mathbf{v} \left|\tr\left[ (A^\mathbf{u}_A \otimes A^\mathbf{v}_B) J^{\cN}_{AB} \right] \right| = \log \max_\mathbf{u} \sum_\mathbf{v} \left| W_{\cN}(\mathbf{v}|\mathbf{u}) \right|,
\end{equation}
where $J^{\cN}_{AB}$ is the Choi matrix for the map $\cN$.
\end{definition}
The calculation of mana can be formulated as finding maximal value over the vector $w(\mathbf{u}) = \sum_\mathbf{v} | W_{\cN}(\mathbf{v}|\mathbf{u})|$ in terms of index $\mathbf{u}$, which can be computed by linear programming, 
\begin{equation}
    \min_z\;  z, \quad \text{s.t.}\;  z\geq w(\mathbf{u}), \; \forall \mathbf{u}.  
\end{equation}
Thus the mana can be computed by the following semidefinite program:
\begin{equation}
\begin{aligned}
\min_{ z} &\;\;  z\\
 {\rm s.t.}
&\; \tr_{B}[J^{\cN}_{AB}] = I_{A}, \;J^{\cN}_{AB} \geq 0, \\
&\;  z \geq \frac{1}{d_B} \sum_{\mathbf{v}} \left| \tr\Big[\left(A^{\mathbf{u}}_{A})\ox A^{\mathbf{v}}_{B}\right)J^{\cN}_{AB}\Big] \right|, \; \forall \mathbf{u},
\end{aligned}
\end{equation}
then $\cM(\cN) = \log z$. This definition can be naturally extended to trace-non-increasing maps by
applying the same Choi operator to the successful branch, which is the
setting relevant to probabilistic universal purification.

\subsection{The robustness of magic}

For qubit systems, however, the discrete phase space approach cannot be applied
in the same way as in odd dimensions. In particular, one has to either
exclude some Clifford operations from the set of free operations
~\cite{delfosse2015wigner,raussendorf2017contextuality}, or lose compatibility
with tensor products~\cite{raussendorf2020phase}. Therefore, to keep all
multi-qubit stabilizer operations free, we use a definition of magic based on
stabilizer states. We denote  $\mathrm{STAB}_n$ as the set of all $n$-qubit stabilizer states. 
Refs.~\cite{howard2017application, seddon2019quantifying} introduce a scheme to decompose density matrices as real linear combinations of pure stabilizer state projectors.

The pure states in $\mathrm{STAB}_n$ form an overcomplete basis for the set of $2^n$-dimensional density matrices, where $n$ is the number of qubits. Any density matrix can be decomposed as an affine combination of pure stabilizer state projectors, i.e.,
\begin{equation}
    \rho = \sum_j x_j \ketbra{\phi_j}{\phi_j}, \quad \sum_jx_j=1, \quad \ketbra{\phi_j}{\phi_j}\in \mathrm{STAB}_n.
\end{equation}
Robustness of magic is defined as the minimal $l_1$-norm $\|\mathbf{x}\|_1 = \sum_j |x_j|$ over all possible decompositions~\cite{howard2017application}, i.e.,
\begin{equation}
    \cR(\rho) \coloneqq \min_\mathbf{x}\Big\{ \|\mathbf{x}\|_1: \sum_j x_j\ketbra{\phi_j}{\phi_j}=\rho,~\ketbra{\phi_j}{\phi_j}\in \mathrm{STAB}_n\Big\}.
\end{equation}
Equivalently, it is given by
\begin{equation}
    \cR(\rho) = \min_{\rho_{\pm} \in \mathrm{STAB}_n} \big\{2t + 1: (1+t)\rho_+ - t\rho_-=\rho, ~t \geq 0 \big\}.
\end{equation}
The optimization of the robustness can be written in terms of a linear system as~\cite{howard2017application}
\begin{equation}\label{eq: LP_robustness of magic}
    \cR(\rho) = \min_{\mathbf{x}} \big\{\|\mathbf{x}\|_1:   G\mathbf{x}=b \big\},
\end{equation}
where $G_{ij} = \tr [P_i \ketbra{\phi_j}{\phi_j}]$, $\ketbra{\phi_j}{\phi_j} \in \mathrm{STAB}_n$, $b_i=\tr[P_i \rho]$, and $P_i$ is the $i$-th Pauli operator for the given number of qubits.
The robustness of magic is faithful: if \(\rho\) lies in the stabilizer
polytope, then \(\cR(\rho)=1\), and vice versa. For a positive semidefinite
unnormalized operator \(\rho\) with \(\tr\rho=p\), one has $\cR(\rho)=p$
if and only if \(\rho/p\) lies in the stabilizer polytope. 


For deterministic quantum operations, robustness of magic can also be extended
to a channel robustness, where the free operations are completely
stabilizer-preserving operations~\cite{seddon2019quantifying}. We briefly
recall this notion for comparison.

\begin{definition}{\rm (Completely stabilizer preserving operations~\cite{seddon2019quantifying})}\label{def:cspo operations} 
A CPTP map $\cN_{A\rightarrow B}$ is called completely stabilizer-preserving if for any system $R$, the following holds
\begin{equation}
    ({\rm \cI}_R \ox \cN_{A\to B})(\rho_{RA})\in {\rm STAB}_{m + n}, \quad \forall \rho_{RA} \in {\rm STAB}_{m + n},
\end{equation}
where $m,n$ are the numbers of qubits for systems $R$ and $A$, respectively. 
\end{definition}
We abbreviate completely stabilizer-preserving operations as CSPOs for convenience. It was proved that a quantum channel is a CSPO if and only if its Choi state $\Phi_{AB}^{\cN} = J_{AB}^{\cN}/d_A$ is a stabilizer state~\cite[Theorem~3.1]{seddon2019quantifying}. The channel robustness of magic $\cR_*(\cN)$ is then defined as 
\begin{equation}
    \cR_*(\cN) = \min_{\Lambda_{\pm}\in \text{CSPO}} \big\{2t+1: (1+t)\Lambda_+ - t\Lambda_-=\cN, ~t\geq 0 \big\}.
\end{equation}
An equivalent definition in terms of the Choi state is
\begin{equation}
    \cR_*(\cN) = \min_{\rho_{\pm} \in \mathrm{STAB}_{2n}} \Big\{2t+1: (1+t)\rho_+ - t\rho_-=\Phi_{AB}^{\cN}, ~t \geq 0,
    ~ \tr_B[\rho_{\pm}]=\frac{I_A}{d_A} \Big\}.
\end{equation}
This quantity can also be computed using the linear program in Eq.~\eqref{eq: LP_robustness of magic}. Channel robustness is also faithful: if $\cN$ is a CSPO, then $\cR_*(\cN)=1$; otherwise $\cR_*(\cN)>1$.

In this work, however, the purification protocol is probabilistic and we focus
on the successful trace-non-increasing branch. Therefore, the main quantity used
below is not the CPTP channel robustness \(\cR_*(\cN)\), but the stabilizer
robustness of the Choi state \(\cR(\Phi^\cN_{AB})\).

\subsection{Tripartite states with partially transposed permutation symmetry}

Positivity in permutation-type algebras is useful in symmetry-reduced quantum information problems. Here we use the partially transposed permutation algebra, which gives a finite-dimensional representation of the walled Brauer algebra. In the present tripartite setting, we consider Hermitian operators \(\omega\) whose partial transpose on the third subsystem belongs to the permutation algebra,
\begin{equation}
    \omega^{T_3}\in \operatorname{span}\{\mathbf P_3(\pi):\pi\in \cS_3\}.
\end{equation}
Equivalently, \(\omega\) belongs to the commutant of
\(U\otimes U\otimes \overline U\).

For our purposes, it is convenient to introduce
\begin{equation}
    \begin{aligned}
    &I\coloneqq \mathbf{P}_3((1)),\quad
X\coloneqq \mathbf{P}_3((23))^{T_3},\quad
V\coloneqq \mathbf{P}_3((12)). \\
&\bP_3((13))^{T_3} = VXV, \quad \bP_3((132))^{T_3} = XV, \quad \bP_3((123))^{T_3} = VX.
\end{aligned}
\end{equation}
These operators satisfy $X^2=dX$, $V^2=I$, $XVX=X$, $\tr X=d^2$, and $\tr[XV]=d$.
Following \cite{eggeling2001separability}, define
\begin{equation}
\begin{aligned}
&S_+ \coloneqq \frac{I+V}{2}\left(I-\frac{2X}{d+1}\right)\frac{I+V}{2},\qquad
S_- \coloneqq \frac{I-V}{2}\left(I-\frac{2X}{d-1}\right)\frac{I-V}{2},\\
&S_0 \coloneqq \frac{1}{d^2-1}\big[d(X+VXV)-(XV+VX)\big],\\
&S_1 \coloneqq \frac{1}{d^2-1}\big[d(XV+VX)-(X+VXV)\big],\\
&S_2 \coloneqq \frac{1}{\sqrt{d^2-1}}(X-VXV),\qquad
S_3 \coloneqq \frac{i}{\sqrt{d^2-1}}(XV-VX).
\end{aligned}
\end{equation}
These matrices give a convenient parametrization of the Hermitian part of the
algebra generated by \(X\) and \(V\).
\begin{lemma}[Positivity criterion~\cite{eggeling2001separability}]\label{lem:EW-positivity}
Let $\omega$ be a Hermitian operator with $\tr[\omega]=1$ in the algebra generated by $X$ and $V$. Define
\begin{equation}
    s_+\coloneqq \tr[\omega S_+],\qquad
s_-\coloneqq \tr[\omega S_-],\qquad
s_j\coloneqq \tr[\omega S_j],\quad j\in\{0,1,2,3\}.
\end{equation}
Then $\omega$ is a density matrix if and only if
\begin{equation}\label{eq:EW-positivity-conditions}
s_+ + s_- + s_0 = 1, \quad s_+\ge 0,\quad
s_-\ge 0,\quad
s_0\ge 0,\quad
s_1^2+s_2^2+s_3^2\le s_0^2.
\end{equation}
\end{lemma}


\section{Formal definitions of mana and robustness of purification as SDPs}\label{app:sdp}

In this appendix, we give the SDP formulations for the two resource quantities used in the main text. We first express the fidelity and success constraints as linear constraints on the Choi operator, and then derive the corresponding primal and dual programs for mana and robustness of purification.

Let $\ket{\psi}$ be the ideal pure state and write $\psi=\ketbra{\psi}{\psi}$. We assume that each copy undergoes the depolarizing channel $\cD_\delta(\cdot)=(1-\delta)(\cdot)+\delta I_d/d$. The fidelity of a single noisy copy with the target is $\lambda_0 \coloneqq \tr[\cD_\delta(\psi)\psi] =1-\frac{d-1}{d}\delta$. Given $n$ noisy copies, a probabilistic protocol is described by a completely positive trace-non-increasing map $\cE_{A^n\to A}$, and the unnormalized output conditioned on success is $\sigma_\psi=\cE_{A^n\to A}(\cD_\delta(\psi)^{\ox n})$~\cite{fiuravsek2004optimal}. For a finite test set $\Psi=\{\psi_i\}_{i=1}^{|\Psi|}$, the average purification fidelity of a protocol is
\begin{equation}\label{eq:max_fidelity}
    F(\cE;\Psi)
    =
    \frac{\sum_{\psi_i\in\Psi}\tr[\sigma_{\psi_i}\psi_i]}
    {\sum_{\psi_i\in\Psi}\tr[\sigma_{\psi_i}]},
\end{equation}
where $\sigma_{\psi_i}=\cE_{A^n\to A}(\cD_\delta(\psi_i)^{\ox n})$
is the unnormalized postselected state. The corresponding average success
probability is
\begin{equation}
    P(\cE;\Psi)
    =
    \frac{1}{|\Psi|}
    \sum_{\psi_i\in\Psi}\tr[\sigma_{\psi_i}].
\end{equation}
For a prescribed target pair \((p,f)\), we require $F(\cE;\Psi)=f,\;P(\cE;\Psi)=p.$ We focus on the nontrivial purification regime \(f\ge\lambda_0\), where the
conditional output fidelity improves on the original noisy copy.

For the fixed test set \(\Psi\) and depolarizing channel \(\cD_\delta\), let
\begin{equation}
\begin{aligned}
\mathcal A_\delta(n,f,p;\Psi)
\coloneqq
\Big\{
\cE_{A^n\to A}:\;
\cE \ {\rm \in \ CPTN},
F(\cE;\Psi)=f,\;
P(\cE;\Psi)=p
\Big\}.
\end{aligned}
\end{equation}
This is the set of probabilistic purification operations that attain the
prescribed success probability and fidelity.

For odd-dimensional qudit systems, we quantify the magic of the operation
conditioned on success by its mana
~\cite{gross2006hudson,veitch2014resource,wang2019quantifying}.

\begin{definition}\label{def:purification_mana}
{\rm (Mana of purification for given fidelity and probability)}
For target fidelity \(\lambda_0\le f<1\), success probability
\(p\in(0,1]\), and \(n\) noisy copies, the mana of purification is
\begin{equation}
\begin{aligned}
        \cM_{\cD_\delta}(n,f,p;\Psi)
    \coloneqq \log
    \min_{\cE\in\mathcal A_\delta(n,f,p;\Psi)}
    \max_{\mathbf u}
    \frac{1}{p}
    \sum_{\mathbf v}
    \left|W_{\cE}(\mathbf v|\mathbf u)\right|.
\end{aligned}
\end{equation}
\end{definition}

The fidelity and success probability constraints can be written compactly as linear constraints on the Choi operator. To this end, define
\begin{equation}\label{eq:discrete Q and R}
\begin{aligned}
     Q_{A^n_IA_O} &\coloneqq \frac{1}{|\Psi|}\sum_{\psi_i\in \Psi}\cD_{\delta}(\psi_i)^{\ox n} \ox \psi_i = (\cD_{\delta}^{\ox n} \ox \cI)\Big(\frac{1}{|\Psi|}\sum_{\psi_i\in \Psi}\psi_i^{\ox {n+1}}\Big),\\
     R_{A^n_IA_O} &\coloneqq \cD_{\delta}^{\ox n} \Big(\frac{1}{|\Psi|}\sum_{\psi_i\in \Psi}\psi_i^{\ox n} \Big)\ox I,
\end{aligned}
\end{equation}

\begin{proposition}\label{sdp: mana of purification}
    For a finite set of pure states $\Psi = \{\psi_i\}_{i=1}^{|\Psi|} \subset \cD(\cH_A)$ under a $d$-dimensional depolarizing channel $\cD_\delta$, the mana of purification with feasible target pair \((p,f)\) and \(n\) noisy copies is computed by the following SDPs:
    \begin{equation}
\begin{aligned}
&\underline{\textbf{Primal Program}}\\
\min_{J^{\cE}_{A^n_IA_O}, z} &\;\;  z/p\\
 {\rm s.t.}
&\;  \tr\Big[J^{\cE}_{A^n_IA_O}Q_{A^n_IA_O}^{T_{A^n_I}}\Big] = pf \\
&\; \tr\Big[J^{\cE}_{A^n_IA_O}R_{A^n_IA_O}^{T_{A^n_I}}\Big]=p, \\
&\; \tr_{A_O}[J^{\cE}_{A^n_IA_O}]\leq I_{A^n_I}, \;J^{\cE}_{A^n_IA_O}\geq 0, \\
&\; z \geq
\frac{1}{d}
\sum_{\mathbf v}
\left|
\tr\Big[
\left(A^{\mathbf u}_{A_I^n}\otimes A^{\mathbf v}_{A_O}\right)
J^\cE_{A_I^nA_O}
\Big]
\right|, \\
&\qquad \forall \mathbf u .
\end{aligned}
\begin{aligned}
&\underline{\textbf{Dual Program}}\\
\max_{\alpha,\beta,Y_{A^n_I},S_{\mathbf{u},\mathbf{v}},\gamma_{\mathbf{u}}}
&  pf \alpha + p \beta - \tr[Y_{A^n_I}] \\
{\rm s.t.}\;
& Y_{A^n_I}\otimes I_{A_O}
+ C_{A^n_I A_O}
\ge \alpha Q_{A^n_IA_O}^{T_{A^n_I}}
+ \beta R_{A^n_IA_O}^{T_{A^n_I}}, \\
& C_{A^n_I A_O} = \sum_{\mathbf{u},\mathbf{v}}
\frac{S_{\mathbf{u},\mathbf{v}}}{d}
\left(A^{\mathbf{u}}_{A^n_I} \otimes A^{\mathbf{v}}_{A_O}\right) \\
& |S_{\mathbf{u},\mathbf{v}}| \le \gamma_{\mathbf{u}},
\; \forall \mathbf{u},\mathbf{v}, \\ 
& \sum_{\mathbf{u}} \gamma_{\mathbf{u}} \le \frac{1}{p}, \; Y_{A^n_I}\ge 0,\quad \gamma_{\mathbf{u}}\ge 0.
\end{aligned}
\end{equation}
where $J^{\cE}_{A^n_IA_O}$ denotes the Choi operator of $\cE_{A^n \to A}$, $T_{A^n_I}$ denotes the partial transpose operation on system $A^n_I$, and $Q_{A^n_IA_O}, R_{A^n_IA_O}$ are given in Eq.~\eqref{eq:discrete Q and R}.  $A^{\mathbf{u}}_{A_I^n}$ and 
$A^{\mathbf{v}}_{A_O}$ denote the $n$-copy and 1-copy phase space point operators, respectively. The mana of purification is then the logarithm of the optimal value of these programs.
\end{proposition}

\begin{proof}
We first derive the primal SDP from the Definition~\ref{def:purification_mana}, then derive its dual program.
The objective function of the primal SDP is given by the definition of the mana of a quantum channel in Def.~\ref{def: channel mana} after linearizing the absolute values. The first constraint is calculated by
\begin{equation}\label{eq:objective discrete sdp}
\begin{aligned}
    \frac{1}{p|\Psi|}\sum_{\psi_i \in \Psi}\tr[\sigma_{\psi_i} \psi_i] &= \frac{1}{p|\Psi|}\sum_{\psi_i \in \Psi} \tr\big[\cE_{A^n \to A}(\cD_\delta(\psi_i)^{\ox n}) \psi_i \big] \\
    &= \frac{1}{p|\Psi|}\sum_{\psi_i \in \Psi} \tr\Big[(J^{\cE}_{A^n_IA_O})^{T_{A_I^n}} \big(\cD_\delta(\psi_i)^{\ox n} \ox \psi_i\big)\Big] \\
    &= \frac{1}{p}\tr \left[(J^{\cE}_{A^n_IA_O})^{T_{A_I^n}}\frac{1}{|\Psi|}\sum_{\psi_i \in \Psi} \cD_\delta(\psi_i)^{\ox n} \ox \psi_i  \right] \\
    &= \frac{1}{p}\tr\Big[J^{\cE}_{A^n_IA_O}Q_{A^n_IA_O}^{T_{A_I^n}}\Big]\\
    &= f, 
\end{aligned}
\end{equation}
and similarly $\frac{1}{|\Psi|} \sum_{\psi_i \in \Psi}\tr[\sigma_{\psi_i}] = \tr\Big[J^{\cE}_{A^n_IA_O} R_{A^n_IA_O}^{T_{A_I^n}}\Big] = p$. Furthermore, the constraints are derived from the calculation of mana for the map $\cE_{A^n \to A}$ where $\tr_{A_O}[J^{\cE}_{A^n_IA_O}]\leq I_{A^n_I},\;J^{\cE}_{A^n_IA_O}\geq 0$ imply the Choi matrix of a CPTN map. This completes the primal program.

Introducing auxiliary variables $t_{\mathbf{u},\mathbf{v}}$ to linearize the absolute value, the primal becomes
\begin{equation}
 \begin{aligned}
    \min_{J^{\cE}_{A^n_IA_O},  z} &\;\; z/p\\
{\rm s.t.}\quad
&\;  \tr\Big[J^{\cE}_{A^n_IA_O}Q_{A^n_IA_O}^{T_{A^n_I}}\Big] = pf, \quad \tr\Big[J^{\cE}_{A^n_IA_O}R_{A^n_IA_O}^{T_{A^n_I}}\Big]=p, \\
&\; \tr_{A_O}[J^{\cE}_{A^n_IA_O}]\leq I_{A^n_I}, \; J^{\cE}_{A^n_IA_O}\geq 0, \; z \ge 0,\\
&\; \frac{1}{d}\tr\Big[\left(A^{\mathbf{u}}_{A^n_I} \otimes A^{\mathbf{v}}_{A_O}\right)J^{\cE}_{A^n_IA_O}\Big]- t_{\mathbf{u},\mathbf{v}} \le 0,
\quad \forall \mathbf{u},\mathbf{v}, \\
&\; -\frac{1}{d}\tr\Big[\left(A^{\mathbf{u}}_{A^n_I} \otimes A^{\mathbf{v}}_{A_O}\right)J^{\cE}_{A^n_IA_O}\Big]- t_{\mathbf{u},\mathbf{v}} \le 0,
\quad \forall \mathbf{u},\mathbf{v}, \\
&\; \sum_{\mathbf{v}} t_{\mathbf{u},\mathbf{v}} - z \le 0,\quad \forall \mathbf{u}.
\end{aligned}   
\end{equation}
Let $\alpha,\beta\in\mathbb{R}$ be the Lagrange multipliers for the two equality constraints, let
$Y_{A^n_I}\geq 0$ be the multiplier for $\tr_{A_O}[J^{\cE}_{A^n_IA_O}] \leq I_{A^n_I}$, let
$\lambda^{\pm}_{\mathbf{u},\mathbf{v}}\geq 0$ be the multipliers for the two inequalities involving
$t_{\mathbf{u},\mathbf{v}}$, and let $\gamma_{\mathbf{u}}\geq 0$ be the multiplier for
$\sum_{\mathbf{v}} t_{\mathbf{u},\mathbf{v}} - z \le 0$.

The Lagrangian is
\begin{equation}
 \begin{aligned}
L
&= z/p
+ \alpha \left(pf-\tr \Big[J^{\cE}_{A^n_IA_O}Q_{A^n_IA_O}^{T_{A^n_I}}\Big]\right)
+ \beta \left(p-\tr \Big[J^{\cE}_{A^n_IA_O}R_{A^n_IA_O}^{T_{A^n_I}}\Big]\right) \\
& + \tr \Big[Y_{A^n_I}\big(\tr_{A_O}[J^{\cE}_{A^n_IA_O}] - I_{A^n_I}\big)\Big] \\
& + \sum_{\mathbf{u},\mathbf{v}} \lambda^+_{\mathbf{u},\mathbf{v}}
\left(
\frac{1}{d}\tr \Big[\left(A^{\mathbf{u}}_{A^n_I}\otimes A^{\mathbf{v}}_{A_O}\right)J^{\cE}_{A^n_IA_O}\Big]
- t_{\mathbf{u},\mathbf{v}}
\right) \\
& + \sum_{\mathbf{u},\mathbf{v}} \lambda^-_{\mathbf{u},\mathbf{v}}
\left(
-\frac{1}{d}\tr \Big[\left(A^{\mathbf{u}}_{A^n_I}\otimes A^{\mathbf{v}}_{A_O}\right)J^{\cE}_{A^n_IA_O}\Big]
- t_{\mathbf{u},\mathbf{v}}
\right) \\
& + \sum_{\mathbf{u}} \gamma_{\mathbf{u}}
\left(\sum_{\mathbf{v}} t_{\mathbf{u},\mathbf{v}} - z\right) \\
&= pf \alpha + p \beta - \tr[Y_{A^n_I}]  + z\left(\frac{1}{p} - \sum_{\mathbf{u}}\gamma_{\mathbf{u}}\right)  + \sum_{\mathbf{u},\mathbf{v}} t_{\mathbf{u},\mathbf{v}}
\left(-\lambda^+_{\mathbf{u},\mathbf{v}}-\lambda^-_{\mathbf{u},\mathbf{v}}+\gamma_{\mathbf{u}}\right) \\
& + \tr\Bigg[
J^{\cE}_{A^n_IA_O}
\Bigg(
-\alpha Q_{A^n_IA_O}^{T_{A^n_I}}
-\beta R_{A^n_IA_O}^{T_{A^n_I}}
+ Y_{A^n_I}\otimes I_{A_O} 
+ \sum_{\mathbf{u},\mathbf{v}}
\frac{\lambda^+_{\mathbf{u},\mathbf{v}}-\lambda^-_{\mathbf{u},\mathbf{v}}}{d}
\left(A^{\mathbf{u}}_{A^n_I}\otimes A^{\mathbf{v}}_{A_O}\right)
\Bigg)
\Bigg].
\end{aligned}   
\end{equation}
For the dual function to be finite, the coefficients of the free variables $z$ and $t_{\mathbf{u},\mathbf{v}}$ must satisfy
\begin{align}
\sum_{\mathbf{u}} \gamma_{\mathbf{u}} \le \frac{1}{p},
\qquad
\lambda^+_{\mathbf{u},\mathbf{v}}+\lambda^-_{\mathbf{u},\mathbf{v}} = \gamma_{\mathbf{u}},
\quad \forall \mathbf{u},\mathbf{v},
\end{align}
where the inequality $\sum_{\mathbf{u}} \gamma_{\mathbf{u}} \le 1/p$ comes from the explicit constraint
$z\ge 0$ in the primal. The coefficient of $J^{\cE}_{A^n_IA_O}$ must obey
\begin{align}
-\alpha Q_{A^n_IA_O}^{T_{A^n_I}}
-\beta R_{A^n_IA_O}^{T_{A^n_I}}
+ Y_{A^n_I}\otimes I_{A_O}
+ \sum_{\mathbf{u},\mathbf{v}}
\frac{\lambda^+_{\mathbf{u},\mathbf{v}}-\lambda^-_{\mathbf{u},\mathbf{v}}}{d}
\left(A^{\mathbf{u}}_{A^n_I}\otimes A^{\mathbf{v}}_{A_O}\right)
\ge 0.
\end{align}

Therefore the dual SDP is
\begin{equation}
 \begin{aligned}
\max_{\alpha,\beta,Y,\lambda^\pm,\gamma}\quad
& pf \alpha + p \beta - \tr[Y_{A^n_I}] \\
{\rm s.t.}\quad
& -\alpha Q_{A^n_IA_O}^{T_{A^n_I}}
-\beta R_{A^n_IA_O}^{T_{A^n_I}}
+ Y_{A^n_I}\otimes I_{A_O} 
+ \sum_{\mathbf{u},\mathbf{v}}
\frac{\lambda^+_{\mathbf{u},\mathbf{v}}-\lambda^-_{\mathbf{u},\mathbf{v}}}{d}
\left(A^{\mathbf{u}}_{A^n_I}\otimes A^{\mathbf{v}}_{A_O}\right)
\ge 0, \\
& \lambda^+_{\mathbf{u},\mathbf{v}}+\lambda^-_{\mathbf{u},\mathbf{v}}=\gamma_{\mathbf{u}},
\quad \forall \mathbf{u},\mathbf{v}, \\
& \sum_{\mathbf{u}} \gamma_{\mathbf{u}} \le \frac{1}{p}, \\
& Y_{A^n_I}\ge 0,\quad
\lambda^+_{\mathbf{u},\mathbf{v}}\ge 0,\quad
\lambda^-_{\mathbf{u},\mathbf{v}}\ge 0,\quad
\gamma_{\mathbf{u}}\ge 0.
\end{aligned}   
\end{equation}

Defining
\begin{align}
S_{\mathbf{u},\mathbf{v}} := \lambda^+_{\mathbf{u},\mathbf{v}}-\lambda^-_{\mathbf{u},\mathbf{v}},
\end{align}
the relation
\(
\lambda^+_{\mathbf{u},\mathbf{v}}+\lambda^-_{\mathbf{u},\mathbf{v}}=\gamma_{\mathbf{u}}
\)
is equivalent to
\begin{align}
|S_{\mathbf{u},\mathbf{v}}| \le \gamma_{\mathbf{u}},
\quad \forall \mathbf{u},\mathbf{v}.
\end{align}
Hence an equivalent compact form of the dual is
\begin{equation}
 \begin{aligned}
\max_{\alpha,\beta,Y_{A^n_I},S_{\mathbf{u},\mathbf{v}},\gamma_{\mathbf{u}}}\quad
& pf \alpha + p \beta - \tr[Y_{A^n_I}] \\
{\rm s.t.}\quad
& 
Y_{A^n_I}\otimes I_{A_O}
+ C_{A^n_IA_O}
\ge \alpha Q_{A^n_IA_O}^{T_{A^n_I}}
+\beta R_{A^n_IA_O}^{T_{A^n_I}}, \\
& C_{A^n_IA_O} = \sum_{\mathbf{u},\mathbf{v}}
\frac{S_{\mathbf{u},\mathbf{v}}}{d}
\left(A^{\mathbf{u}}_{A^n_I}\otimes A^{\mathbf{v}}_{A_O}\right) \\
& |S_{\mathbf{u},\mathbf{v}}| \le \gamma_{\mathbf{u}},
\quad \forall \mathbf{u},\mathbf{v}, \\
& \sum_{\mathbf{u}} \gamma_{\mathbf{u}} \le \frac{1}{p}, \\
& Y_{A^n_I}\ge 0,\quad \gamma_{\mathbf{u}}\ge 0.
\end{aligned}   
\end{equation}
Taking the logarithm of the optimal objective value gives the mana. Strong duality holds by Slater's condition for the dual SDP; for example, take $\alpha = 0$, $\beta = 0$, $S_{\mathbf{u},\mathbf{v}} = 0$, $Y_{A^n_I} = I_{A^n_I}$, and $\gamma_{\mathbf{u}}= \frac{1}{2d^{2n} p}$ for all $\mathbf{u}$. This completes the proof. 
\end{proof}

The robustness formulation follows the same purification constraints. The only change is that the Wigner-function norm is replaced by a stabilizer decomposition of the Choi state. For multi-qubit systems, we use the stabilizer robustness of the Choi state of
the corresponding trace-non-increasing operation
~\cite{howard2017application,seddon2019quantifying}.

\begin{definition}\label{def:purification_robustness}
{\rm (Robustness of purification for given fidelity and probability)}
For target fidelity \(\lambda_0\le f<1\), success probability
\(p\in(0,1]\), and \(n\) noisy copies, the robustness of purification is
\begin{equation}
    \cR_{\cD_\delta}(n,f,p;\Psi)
    \coloneqq
    \min_{\cE\in\mathcal A_\delta(n,f,p;\Psi)}
    \frac{1}{p}
    \cR(\Phi^\cE_{A_I^nA_O}) ,
\end{equation}
\end{definition}
where $\Phi^\cE_{A_I^nA_O}=J^\cE_{A_I^nA_O}/d_{A_I^n}$ is the Choi state of $\cE_{A^n \to A}$.

\begin{proposition}\label{sdp: robustness of purification}
For a finite set of pure states $\Psi = \{\psi_i\}_{i=1}^{|\Psi|}\subset\cD(\cH_A)$ on an \(m\)-qubit system \(A\) under the depolarizing channel $\cD_\delta$, the robustness of purification with feasible target pair \((p,f)\) and \(n\) noisy copies is computed by the following primal and dual SDPs:
\begin{equation}
\begin{aligned}
&\underline{\textbf{Primal Program}}\\
\min_{J^{\cE}_{A^n_IA_O}, \mathbf{x}} &\;\; \|\mathbf{x}\|_1 /p \\
{\rm s.t.}
&\; \tr\Big[J^{\cE}_{A^n_IA_O}Q_{A^n_IA_O}^{T_{A^n_I}}\Big] = pf \\
&\; \tr\Big[J^{\cE}_{A^n_IA_O}R_{A^n_IA_O}^{T_{A^n_I}}\Big]=p, \\
&\; \tr_{A_O}[J^{\cE}_{A^n_IA_O}]\leq I_{A^n_I},\quad J^{\cE}_{A^n_IA_O}\geq 0, \\
&\; \sum_j G_{ij}x_j=b_i,\qquad \forall i,\\
&\; b_i =  \frac{1}{d_{A^n_I}}\tr\big[P_iJ^{\cE}_{A^n_IA_O}\big],\qquad \forall i.
\end{aligned}
\begin{aligned}
&\underline{\textbf{Dual Program}}\\
\max_{\alpha,\beta,Y_{A^n_I},\mathbf{y}}\quad
& pf\alpha + p \beta - \tr[Y_{A^n_I}] \\
{\rm s.t.}\quad
& Y_{A^n_I}\otimes I_{A_O}
+D_{A^n_IA_O}
\ge \alpha Q_{A^n_IA_O}^{T_{A^n_I}}
+\beta R_{A^n_IA_O}^{T_{A^n_I}}, \\
& D_{A^n_IA_O}=\frac{1}{d_{A^n_I}}\sum_i y_iP_i,\\
& \left|\sum_i y_iG_{ij}\right|\le \frac{1}{p},\qquad \forall j, \\
& Y_{A^n_I}\ge 0.
\end{aligned}
\end{equation}
Here $J^{\cE}_{A^n_IA_O}$ denotes the Choi operator of $\cE_{A^n\to A}$, $T_{A^n_I}$ denotes the partial transpose on $A^n_I$, and $Q_{A^n_IA_O},R_{A^n_IA_O}$ are given in Eq.~\eqref{eq:discrete Q and R}. Moreover, $G_{ij}=\tr[P_i\ketbra{\phi_j}{\phi_j}]$, $\ketbra{\phi_j}{\phi_j}\in\mathrm{STAB}_{(n+1)m}$, and $\{P_i\}$ is a Pauli basis on $A_I^nA_O$.
\end{proposition}

\begin{proof}
The fidelity and success probability constraints are identical to those in Proposition~\ref{sdp: mana of purification}. The constraints $\tr_{A_O}[J^{\cE}_{A^n_IA_O}]\leq I_{A^n_I}$ and $J^{\cE}_{A^n_IA_O}\geq 0$ impose that $J^{\cE}_{A^n_IA_O}$ is the Choi operator of a CPTN map. The remaining linear system is Eq.~\eqref{eq: LP_robustness of magic} applied to the Choi state $\Phi^\cE_{A_I^nA_O}=J^\cE_{A_I^nA_O}/d_{A_I^n}$, so that $b_i=\tr[P_i\Phi^\cE_{A_I^nA_O}]$. This gives the primal program.

To derive the dual, linearize the $\ell_1$-norm by introducing variables $t_j$ satisfying
\begin{equation}
    x_j-t_j\le 0,\qquad -x_j-t_j\le 0,
    \qquad \forall j.
\end{equation}
The primal is equivalently
\begin{equation}
\begin{aligned}
\min_{J^{\cE}_{A^n_IA_O}, \mathbf{x}, \mathbf{t}} &\;\; \frac{1}{p}\sum_j t_j \\
{\rm s.t.}
&\; \tr\Big[J^{\cE}_{A^n_IA_O}Q_{A^n_IA_O}^{T_{A^n_I}}\Big] = pf,
\quad \tr\Big[J^{\cE}_{A^n_IA_O}R_{A^n_IA_O}^{T_{A^n_I}}\Big]=p, \\
&\; \tr_{A_O}[J^{\cE}_{A^n_IA_O}]\leq I_{A^n_I},\quad J^{\cE}_{A^n_IA_O}\geq 0, \\
&\; \frac{1}{d_{A^n_I}}\tr\big[P_iJ^{\cE}_{A^n_IA_O}\big]-\sum_jG_{ij}x_j=0,
\qquad \forall i,\\
&\; x_j-t_j\le 0,\quad -x_j-t_j\le 0,
\qquad \forall j.
\end{aligned}
\end{equation}
Let $\alpha,\beta\in\mathbb{R}$ be the multipliers for the two equality constraints, $Y_{A^n_I}\ge0$ the multiplier for the trace-non-increasing constraint, $y_i\in\mathbb{R}$ the multipliers for the robustness decomposition constraints, and $\lambda_j^\pm\ge0$ the multipliers for the two inequalities involving $x_j$ and $t_j$. The Lagrangian is
\begin{equation}
    \begin{aligned}
L
=&\; \frac{1}{p}\sum_j t_j
+\alpha\left(pf-\tr\Big[J^{\cE}_{A^n_IA_O}Q_{A^n_IA_O}^{T_{A^n_I}}\Big]\right)
+\beta\left(p-\tr\Big[J^{\cE}_{A^n_IA_O}R_{A^n_IA_O}^{T_{A^n_I}}\Big]\right) \\
&\; +\tr\Big[Y_{A^n_I}\big(\tr_{A_O}[J^{\cE}_{A^n_IA_O}]-I_{A^n_I}\big)\Big] \\
&\; +\sum_i y_i\left(\frac{1}{d_{A^n_I}}\tr\big[P_iJ^{\cE}_{A^n_IA_O}\big]-\sum_jG_{ij}x_j\right) \\
&\; +\sum_j\lambda_j^+(x_j-t_j)+\sum_j\lambda_j^-(-x_j-t_j).
\end{aligned}
\end{equation}
Rearranging gives
\begin{equation}
    \begin{aligned}
L
=&\; pf\alpha+p\beta-\tr[Y_{A^n_I}]
+\sum_jt_j\left(\frac{1}{p}-\lambda_j^+-\lambda_j^-\right) \\
&\; +\sum_jx_j\left(\lambda_j^+-\lambda_j^- -\sum_i y_iG_{ij}\right) \\
&\; +\tr\Bigg[J^{\cE}_{A^n_IA_O}\Bigg(
-\alpha Q_{A^n_IA_O}^{T_{A^n_I}}
-\beta R_{A^n_IA_O}^{T_{A^n_I}}
+Y_{A^n_I}\otimes I_{A_O}
+\frac{1}{d_{A^n_I}}\sum_i y_iP_i
\Bigg)\Bigg].
\end{aligned}
\end{equation}
For the dual function to be finite, the coefficients of the free variables $t_j$ and $x_j$ must satisfy
\begin{equation}
\lambda_j^+ + \lambda_j^- = \frac{1}{p},
\qquad
\lambda_j^+ - \lambda_j^- = \sum_i y_iG_{ij},
\qquad \forall j,
\end{equation}
which is equivalent to $|\sum_i y_iG_{ij}|\le1/p$ for every $j$. Minimization over $J^{\cE}_{A^n_IA_O}\ge0$ gives the semidefinite constraint
\begin{equation}
Y_{A^n_I}\otimes I_{A_O}
+\frac{1}{d_{A^n_I}}\sum_i y_iP_i
\ge
\alpha Q_{A^n_IA_O}^{T_{A^n_I}}
+\beta R_{A^n_IA_O}^{T_{A^n_I}}.
\end{equation}
Hence the dual SDP is given by
\begin{equation}
    \begin{aligned}
\max_{\alpha,\beta,Y_{A^n_I},\mathbf{y}}\quad
& pf\alpha + p \beta - \tr[Y_{A^n_I}] \\
{\rm s.t.}\quad
& Y_{A^n_I}\otimes I_{A_O}
+D_{A^n_IA_O}
\ge \alpha Q_{A^n_IA_O}^{T_{A^n_I}}
+\beta R_{A^n_IA_O}^{T_{A^n_I}}, \\
& D_{A^n_IA_O}=\frac{1}{d_{A^n_I}}\sum_i y_iP_i,\\
& \left|\sum_i y_iG_{ij}\right|\le \frac{1}{p},\qquad \forall j, \\
& Y_{A^n_I}\ge 0.
\end{aligned}
\end{equation}
Strong duality holds by Slater's condition for the dual SDP; for example, take
\(\alpha=\beta=0\), \(\mathbf y=0\), and \(Y=I_{A_I^n}\).
Hence strong duality follows from Slater's condition. This completes the proof.
\end{proof}

The universal case is obtained by replacing the finite averages in
\(Q_{A_I^nA_O}\) and \(R_{A_I^nA_O}\) with the corresponding Haar integrals over pure states.


\section{Proof of Theorem~\ref{thm:mana-universal}}
\label{app:mana-proof}

\noindent\textbf{Theorem 1}{\rm (Mana law for universal purification)}
\textit{For two-to-one universal probabilistic purification of odd-dimensional
qudit states under depolarizing noise, the exponentiated mana of universal
purification is exactly linear in the fidelity gain \(f-\lambda_0\).}

\bigskip
\begin{proof}
More precisely, we prove that
\begin{equation}
    \cM_{\cD_\delta}(2,f,p,d)
    =
    \log\left(1+\cK_\cM(f-\lambda_0)\right),
\end{equation}
where
\begin{equation}
\cK_\cM
=
\frac{d+\delta(2-\delta)}
{\lambda_0(1-\delta)\delta},
\end{equation}
and \(\lambda_0=1-\frac{d-1}{d}\delta\) is the fidelity without purification.

We prove the theorem using the SDP in
Proposition~\ref{sdp: mana of purification}. We first give an explicit
feasible construction for the primal SDP, showing that
\begin{equation}
    \cM_{\cD_\delta}(2,f,p,d)
    \le
    \log\left(1+\cK_\cM(f-\lambda_0)\right).
\end{equation}
We then construct a feasible dual solution attaining the same value, which
gives the matching lower bound.

For two-copy universal purification, the operators \(Q_{A^2_IA_O}\) and \(R_{A^2_IA_O}\) are
\begin{equation}\label{eq: Q and R decomposition}
\begin{aligned}
   Q_{A^2_IA_O} & = \int \cD_{\delta}(\psi)^{\ox 2} \ox \psi\, d\psi = \int\left((1-\delta)\psi + \frac{\delta}{d}I \right)^{\ox 2}\ox\psi\,d\psi \\
   &= \left((1-\delta)^2\frac{\Pi_{123}}{D(3,d)} + \frac{(1-\delta)\delta}{d }\frac{(\Pi_{13} \ox I_2 + \Pi_{23} \ox I_1)}{D(2,d)}  + \frac{\delta^2}{d^3}I_{123} \right),\\
   &= \left( \frac{(1-\delta)^2}{(2+d)(1+d)d} + \frac{2(1-\delta)\delta}{d^2(1+d)} + \frac{\delta^2}{d^3} \right)I_{123} \\
   &+ \left(\frac{(1-\delta)^2}{(2+d)(1+d)d} + \frac{(1-\delta)\delta}{d^2(1+d)} \right) \left(\mathbf{P}_3((13)) + \mathbf{P}_3((23))\right) \\
   & + \frac{(1-\delta)^2}{(2+d)(1+d)d} \left( \mathbf{P}_3((12)) + \mathbf{P}_3((123)) + \mathbf{P}_3((132)) \right) \\
   &= \frac{d^2+4d\delta +(2-d)\delta^2}{(d+2)(d+1)d^3}I_{123} + \frac{(1-\delta)(d+2\delta)}{(d+2)(d+1)d^2}  \left(\mathbf{P}_3((13)) + \mathbf{P}_3((23))\right) \\
   & + \frac{(1-\delta)^2}{(d+2)(d+1)d} \left( \mathbf{P}_3((12)) + \mathbf{P}_3((123)) + \mathbf{P}_3((132)) \right) \\
   R_{A^2_IA_O} & = \int\cD_\delta(\psi)^{\ox 2}\ox I_3\,d\psi = \int\left((1-\delta)\psi + \frac{\delta}{d}I \right)^{\ox 2}\ox I_3\,d\psi \\
   &= \left((1-\delta)^2\frac{\Pi_{12}}{D(2,d)} + \frac{2\delta-\delta^2}{d^2} I_{12} \right)\ox I_3\\
   &= \frac{d + \delta(2-\delta)}{(1+d)d^2} I_{123} + \frac{(1-\delta)^2}{(1+d)d}\mathbf{P}_3((12)). 
\end{aligned}
\end{equation}
where $\Pi_{n} = D(n,d) \int \psi^{\ox n} d\psi$ is the projector onto the symmetric subspace of $\cH^{\ox n}_A$ by Schur's lemma~\cite{harrow2013church, khatri2020principles, mele2024introduction}, $D(n,d)=\tbinom{n+d-1}{n}$ is the dimension of the symmetric subspace, and $\mathbf{P}_3((c))$ are permutation operators in the symmetric group $\cS_3$.

The following covariant map gives the primal upper bound:
\begin{equation}\label{eq: Choi matrix parameterized}
    J^{\cE}_{A^2_IA_O} = \mu_1\left(\mathbf{P}_3((13))^{T_3} + \mathbf{P}_3((23))^{T_3} \right) + \mu_2\left(\mathbf{P}_3((123))^{T_3} + \mathbf{P}_3((132))^{T_3} \right),
\end{equation}
where 
\begin{equation}\label{eq: mu1 mu2}
 \begin{aligned}
    \mu_1 &= \frac{((d + \delta - d \delta)^2 - d^2 (\delta - 1)^2 f + 
   d (\delta - 2) \delta f) p}{2 (d-1) (d (\delta - 1) - 
   \delta) (\delta - 1) \delta}, \\
   \mu_2 &= \frac{d (d (\delta + f - 1) - \delta)  p}{2 (d - 1) (d (\delta - 1) - 
   \delta) (\delta - 1) \delta},
\end{aligned}   
\end{equation}
This Choi matrix corresponds to the map given by a linear combination of projections onto the symmetric and antisymmetric subspaces followed by tracing out one output system,
\begin{equation}
    \cE_{A_I^2\to A_O}(\cdot) = \tr_{A_{O, 2}}\left[2(\mu_1 + \mu_2) \Pi_{\mathrm{sym}}(\cdot)\Pi_{\mathrm{sym}}^\dagger + 2(\mu_1 - \mu_2)\Pi_{\mathrm{asym}}(\cdot)\Pi_{\mathrm{asym}}^\dagger \right], 
\end{equation}
where $\Pi_{\mathrm{sym}} = \left(I_2 + \mathbf{P}_2((12))\right)/2$ and $\Pi_{\mathrm{asym}} = \left(I_2 - \mathbf{P}_2((12))\right)/2$.
We will show that this map with $\mu_1, \mu_2$ in Eq.~\eqref{eq: mu1 mu2} satisfies all the constraints. 
The expressions for $\mu_1$ and $\mu_2$ come from solving the following equations:
\begin{equation}\label{eq: Q and R eq}
\tr\Big[J^{\cE}_{A^2_IA_O}Q_{A^2_IA_O}^{T_{A^2_I}}\Big] = pf, \quad \tr\Big[J^{\cE}_{A^2_IA_O}R_{A^2_IA_O}^{T_{A^2_I}}\Big] = p.
\end{equation}
Using Eq.~\eqref{eq: Q and R decomposition}, we have
\begin{equation}
 \begin{aligned}
    \tr\Big[J^{\cE}_{A^2_IA_O}Q_{A^2_IA_O}^{T_{A^2_I}}\Big] &= \tr\Big[\left(J^{\cE}_{A^2_IA_O}\right)^{T_{A^2_I}} Q_{A^2_IA_O}\Big] \\
    &= \tr\Big[ \left[ \mu_1\left(\mathbf{P}_3((13))^{T} + \mathbf{P}_3((23))^{T} \right) + \mu_2\left(\mathbf{P}_3((123))^{T} + \mathbf{P}_3((132))^{T} \right) \right]  Q_{A^2_IA_O} \Big] \\ 
    &= \tr\Big[ \left[ \mu_1\left(\mathbf{P}_3((13)) + \mathbf{P}_3((23)) \right) + \mu_2\left(\mathbf{P}_3((123)) + \mathbf{P}_3((132)) \right) \right]  Q_{A^2_IA_O} \Big] \\ 
    &= \left( \frac{(1-\delta)^2}{(2+d)(1+d)d} + \frac{2(1-\delta)\delta}{d^2(1+d)} + \frac{\delta^2}{d^3} \right)(2d^2 \mu_1 + 2d\mu_2) \\
    &+  \left(\frac{(1-\delta)^2}{(2+d)(1+d)d} + \frac{(1-\delta)\delta}{d^2(1+d)} \right)(2\mu_1(d^3+d) + 4\mu_2 d^2) \\
    &+ \frac{(1-\delta)^2}{(2+d)(1+d)d}(\mu_1(4d^2 + 2d) + 2\mu_2(d^3+d^2+d)) \\
    &= \frac{2(\delta - d (\delta-1)) (\delta \mu_2 + 
   d (\mu_1 + \mu_2 - \delta \mu_2)))}{d^2}\\
    \tr\Big[J^{\cE}_{A^2_IA_O}R_{A^2_IA_O}^{T_{A^2_I}}\Big] &= \tr\Big[ \left[ \mu_1\left(\mathbf{P}_3((13)) + \mathbf{P}_3((23)) \right) + \mu_2\left(\mathbf{P}_3((123)) + \mathbf{P}_3((132)) \right) \right]  R_{A^2_IA_O} \Big] \\
    &= \left(\frac{(1-\delta)^2}{(1+d)d} + \frac{2\delta - \delta^2}{d^2} \right)(2d^2 \mu_1 + 2d\mu_2) + \frac{(1-\delta)^2}{(1+d)d}(2\mu_1d + 2\mu_2d^2)\\
    &= 2 (\mu_1 + \mu_2 + \frac{(d-1) (\delta - 2) \delta \mu_2}{d}), 
\end{aligned}   
\end{equation}
which we rewrite the equations for clarity, 
\begin{equation}
\begin{aligned}
    p &= 2 (\mu_1 + \mu_2 + \frac{((d - 1) (\delta - 2) \delta \mu_2)}{d}) = 2\mu_1 + 2\left(1 + \frac{(d-1)(\delta-2)\delta}{d}\right)\mu_2 \\
    f&= \frac{(\delta - d (\delta - 1)) (\delta \mu_2 + d (\mu_1 + \mu_2 - \delta \mu_2))}{
 d ((2 - \delta) \delta \mu_2 + d (\mu_1 + (\delta - 1)^2 \mu_2))} = \frac{\lambda_0(d\mu_1 + d\lambda_0\mu_2)}{d\mu_1 + m\mu_2},
\end{aligned}    
\end{equation}
where we denote $m = (2-\delta)\delta + d(\delta-1)^2$. 
Solving these equations gives Eq.~\eqref{eq: mu1 mu2}. 

To make $\cE_{A_I^2\to A_O}$ a completely positive map corresponding $J^{\cE}_{A^2_IA_O}\geq 0$, each coefficient of projection should be positive, i.e.
\[    \mu_1 + \mu_2\geq 0, \quad \mu_1 - \mu_2\geq 0,  
\]
which also means $\mu_1\geq |\mu_2| \geq 0$. The trace non-increasing condition gives
\begin{align}
    \tr_{A_O}[J^{\cE}_{A^2_IA_O}] = 2\mu_1 I_{A_I^2} + 2\mu_2 \mathbf{P}((12))_{A_I^2} = 2(\mu_1 + \mu_2)\Pi_{\mathrm{sym}} + 2(\mu_1 - \mu_2)\Pi_{\mathrm{asym}} \leq I_{A_I^2}, 
\end{align}
which indicates that the eigenvalues on $\Pi_{\mathrm{sym}}$ and $\Pi_{\mathrm{asym}}$ should be at most 1, i.e.,
\[
    2(\mu_1 + \mu_2) \le 1, \quad 2(\mu_1 - \mu_2)\le 1.
\]
In the purification regime \(f\ge \lambda_0\), Eq.~\eqref{eq: mu1 mu2} further gives \(\mu_2\ge0\). Hence, for the feasible purification branch considered here,
\begin{equation}\label{eq: mu1 mu2 range}
    0\le \mu_1 + \mu_2 \le \frac{1}{2}, \quad 0\le \mu_1 - \mu_2\le \frac{1}{2}, \quad \mu_1\geq0, \quad \mu_2\ge 0.
\end{equation}
These inequalities describe the golden point behavior of purification proposed in~\cite{yao2024protocols}. Substituting the range of $\mu_1, \mu_2$ in Eq.~\eqref{eq: mu1 mu2 range}, we can observe the behavior of $p$ and $f$. To make this behavior explicit, let
\begin{equation}
    t = \frac{\mu_2}{\mu_1}, \qquad s = 2(\mu_1+\mu_2),
\end{equation}
By Eq.~\eqref{eq: mu1 mu2 range}, we have
\[
    0\leq t\leq 1, \qquad 0\leq s\leq 1.
\]
Conversely,
\begin{equation}
    \mu_1 = \frac{s}{2(1+t)}, \qquad \mu_2 = \frac{st}{2(1+t)}.
\end{equation}
Substituting these expressions into $p$ and $f$, we obtain
\begin{equation}\label{eq: p f depending on s t}
    p = s\,\frac{1+\frac{m}{d}t}{1+t},
    \qquad
    f =\frac{\lambda_0(d+d\lambda_0t)}{d+mt}.
\end{equation}
Therefore, the fidelity depends only on the ratio $t=\mu_2/\mu_1$, while the success probability is linear in the scale parameter $s=2(\mu_1+\mu_2)$. Taking the derivative $\frac{\partial f}{\partial t}$ gives the maximal fidelity
\begin{equation}
    f^\star = \frac{\lambda_0(1+\lambda_0)}{1+\frac{m}{d}}
\end{equation} 
at $\mu_1=\mu_2=\frac{1}{4}$, which gives $t = 1, s = 1$. The corresponding success probability is given by
\begin{equation}
    p^\star = \frac{1 + \frac{m}{d}}{2}.
\end{equation}
For any smaller success probability $p \leq p^\star$, one can move along the ray $\mu_1=\mu_2, t=1$ by rescaling $(\mu_1,
\mu_2)$, so the fidelity remains equal to $f^\star$. When $p > p^\star$, the ratio $t=1$ is no longer feasible, and one must decrease $t$ along the ray $s = 1$, which strictly decreases the fidelity. This behavior agrees with the result in~\cite{yao2024protocols}, and the point $(p^\star, f^\star)$ is called the golden point for the two-copy input purification task.

The above calculation shows that our map in Eq.~\eqref{eq: Choi matrix parameterized} can implement the desired universal purification task with fixed success probability $p$ and fidelity $f$, and our next step is to calculate its mana. To calculate $z$, 
we notice that phase space point operators of a composite system are tensor products of the phase space point operators of their subsystems, $A^{\mathbf{u}}_{A_I^2} = A^{\mathbf{u}_1}_{A_I} \otimes A^{\mathbf{u}_2}_{A_I}$. Ignoring the
system subscripts, we calculate the discrete Wigner function of $J^{\cE}_{A_I^2A_O}$:
\begin{equation}
 \begin{aligned}
    W_\cE(\mathbf{v}|\mathbf{u}) &=  \frac{1}{d}  \tr\left[ (A^{\mathbf{u}_1} \otimes A^{\mathbf{u}_2} \otimes A^\mathbf{v}) J^{\cE}_{A_I^2A_O} \right]  \\
    &= \frac{1}{d} \tr\left[ (A^{\mathbf{u}_1} \otimes A^{\mathbf{u}_2} \otimes A^\mathbf{v}) \left(\mu_1  \left(\mathbf{P}_3((13))^{T_3} + \mathbf{P}_3((23))^{T_3} \right) + \mu_2\left(\mathbf{P}_3((123))^{T_3} + \mathbf{P}_3((132))^{T_3} \right)\right) \right]  \\
    &= \frac{1}{d} \tr\left[ (A^{\mathbf{u}_1} \otimes A^{\mathbf{u}_2} \otimes A^\mathbf{v}) \left(\mu_1  \left(\mathbf{P}_3((13)) + \mathbf{P}_3((23)) \right) + \mu_2\left(\mathbf{P}_3((123)) + \mathbf{P}_3((132)) \right)\right)  \right] ,
\end{aligned}   
\end{equation}
where we use $\{A^\mathbf{u}\}_\mathbf{u} = \{(A^\mathbf{u})^T\}_\mathbf{u}$, thus the transpose can be ignored. Using trace identities for permutation operators and the properties of phase space point operators, we have
\begin{equation}
 \begin{aligned}
    \tr \left[ (A^{\mathbf{u}_1} \otimes A^{\mathbf{u}_2} \otimes A^\mathbf{v})\mathbf{P}_3((13)) \right] &= d\delta(\mathbf{u}_1, \mathbf{v})\\
    \tr \left[ (A^{\mathbf{u}_1} \otimes A^{\mathbf{u}_2} \otimes A^\mathbf{v})\mathbf{P}_3((23)) \right] &= d\delta(\mathbf{u}_2, \mathbf{v}),\\
    \tr \left[ (A^{\mathbf{u}_1} \otimes A^{\mathbf{u}_2} \otimes A^\mathbf{v})\mathbf{P}_3((123)) \right] &= \tr\left[ A^{\mathbf{u}_1} A^{\mathbf{u}_2} A^\mathbf{v} \right] \\
    \tr \left[(A^{\mathbf{u}_1} \otimes A^{\mathbf{u}_2} \otimes A^\mathbf{v})\mathbf{P}_3((132))\right] &= \tr\left[ A^{\mathbf{u}_2} A^{\mathbf{u}_1} A^\mathbf{v} \right] = (\tr\left[ A^{\mathbf{u}_1} A^{\mathbf{u}_2} A^\mathbf{v} \right])^*. \\
\end{aligned}   
\end{equation}

We also use the following useful property: for any phase space point operators $A^{\mathbf{j}} = A^{(a_1^\mathbf{j}, a_2^\mathbf{j})}$, $A^{\mathbf{k}} = A^{(a_1^\mathbf{k}, a_2^\mathbf{k})}$, and $A^\mathbf{l} = A^{(a_1^\mathbf{l}, a_2^\mathbf{l})}$,
\begin{equation}
 \begin{aligned}
        \tr\left[A^{\mathbf{j}} A^{\mathbf{k}}A^\mathbf{l}\right] &= e^{\frac{4\pi i}{d}\left(a_1^\mathbf{j}(a_2^\mathbf{l} - a_2^\mathbf{k}) + a_1^\mathbf{k}(a_2^\mathbf{j} - a_2^\mathbf{l}) + a_1^\mathbf{l}(a_2^\mathbf{k} - a_2^\mathbf{j})\right)} \\
        &= e^{\frac{4\pi i}{d} \cC(\mathbf{j}, \mathbf{k}, \mathbf{l})},
\end{aligned}   
\end{equation}
where we denote $\cC(\mathbf{j}, \mathbf{k}, \mathbf{l}) = \left(a_1^\mathbf{j}(a_2^\mathbf{l} - a_2^\mathbf{k}) + a_1^\mathbf{k}(a_2^\mathbf{j} - a_2^\mathbf{l}) + a_1^\mathbf{l}(a_2^\mathbf{k} - a_2^\mathbf{j})\right)$. 
Thus 
\begin{equation}
 \begin{aligned}
    W_\cE(\mathbf{v}|\mathbf{u}) &=\mu_1(\delta(\mathbf{u}_1, \mathbf{v}) + \delta(\mathbf{u}_2, \mathbf{v})) + \frac{\mu_2}{d}\left(  \tr\left[ A^{\mathbf{u}_1} A^{\mathbf{u}_2} A^\mathbf{v} \right] + \tr\left[ A^{\mathbf{u}_2} A^{\mathbf{u}_1} A^\mathbf{v} \right]\right)\\
    &= \mu_1(\delta(\mathbf{u}_1, \mathbf{v}) + \delta(\mathbf{u}_2, \mathbf{v})) + 2\frac{\mu_2}{d} \mathrm{Re}(\tr[A^{\mathbf{u}_1} A^{\mathbf{u}_2} A^\mathbf{v}])\\
    &= \mu_1(\delta(\mathbf{u}_1, \mathbf{v}) + \delta(\mathbf{u}_2, \mathbf{v})) + 2\frac{\mu_2}{d} \cos\left( \frac{4\pi}{d}\cC(\mathbf{u}_1, \mathbf{u}_2, \mathbf{v}) \right).
\end{aligned}   
\end{equation}
$z$ is calculated as
\begin{equation}
\begin{aligned}
    z &=  \max_{\mathbf{u}} \sum_{\mathbf{v}} |W_\cE(\mathbf{v}|\mathbf{u})|\\
    &=  \max_{\mathbf{u}_1, \mathbf{u}_2} \sum_{\mathbf{v}} \left|\mu_1(\delta(\mathbf{u}_1, \mathbf{v}) + \delta(\mathbf{u}_2, \mathbf{v})) + 2\frac{\mu_2}{d} \cos\left( \frac{4\pi}{d}\cC(\mathbf{u}_1, \mathbf{u}_2, \mathbf{v}) \right) \right|.  
\end{aligned}    
\end{equation}
Since \(\mu_1,\mu_2\ge0\), the choice \(\mathbf u_1=\mathbf u_2=\mathbf u\) makes all terms in the sum nonnegative and gives \(\cC(\mathbf u,\mathbf u,\mathbf v)=0\) for all \(\mathbf v\). This choice attains the maximum, and hence 
\begin{equation}
 \begin{aligned}
     z &=  \max_{\mathbf{u}} \sum_{\mathbf{v}} \left( 2\mu_1\delta(\mathbf{u}, \mathbf{v}) + 2\frac{\mu_2}{d} \right)  =   2\mu_1 + 2d\mu_2.
\end{aligned}   
\end{equation}
Substitute Eq.~\eqref{eq: mu1 mu2} to give $z$ in terms of $f, p$, 
\begin{equation}
\begin{aligned}
    z &=  p\cdot \frac{-\delta^2 - d ( \delta - 2) \delta (f - 1) + d^2 (\delta + f - 1) }{(d (\delta - 1) - \delta) (\delta - 1) \delta} \\
 &=  p \cdot \left(  \frac{d +\delta(2-\delta)}{\lambda_0(1-\delta)\delta} \left(f - \left(1 - \frac{d-1}{d}\delta \right)\right) + 1\right) \\
 &=  p\left( \cK_\cM(f - \lambda_0) + 1 \right),
\end{aligned}    
\end{equation}
where $\cK_\cM = \frac{d +\delta(2-\delta)}{\lambda_0(1-\delta)\delta}$ and $\lambda_0 = 1 - \frac{d-1}{d}\delta$ is the fidelity without purification. Thus the primal construction gives $\cM_{\cD_\delta}(2,f,p,d)\leq \log (z/p) =\log \left( \cK_\cM(f - \lambda_0) + 1 \right)$. 

To prove $\cM_{\cD_\delta}(2,f,p,d)\geq \log \left( \cK_\cM(f - \lambda_0) + 1 \right)$ using the dual SDP, a feasible solution for the dual SDP is 
\begin{equation}
\begin{aligned}
        &\alpha = \frac{\cK_\cM}{p},\; \beta = \frac{-\cK_\cM\lambda_0 + 1}{p}, \; Y_{A^2_I} = 0, \\
        & \gamma_{\mathbf{u}} = \gamma_{\mathbf{u}_1, \mathbf{u}_2} = \frac{1}{d^2 p}\delta(\mathbf{u}_1, \mathbf{u}_2), \; S_{\mathbf{u},\mathbf{v}} = S_{\mathbf{u}_1, \mathbf{u}_2,\mathbf{v}} = \frac{1}{d^2 p}\delta(\mathbf{u}_1, \mathbf{u}_2), \; \forall \mathbf{v}.
\end{aligned}
\end{equation}
This choice gives the dual objective value \(\cK_\cM(f-\lambda_0)+1\). The normalization and absolute value constraints are immediate from the definitions of \(S_{\mathbf u,\mathbf v}\) and \(\gamma_{\mathbf u}\). It remains to verify the semidefinite constraint
\begin{equation}
    -\alpha Q_{A^2_IA_O}^{T_{12}}
-\beta R_{A^2_IA_O}^{T_{12}}
+ C_{A^2_IA_O}
\ge 0,
\end{equation}
with $Y_{A^2_I} = 0$, where $C_{A^2_IA_O}  = \sum_{\mathbf{u},\mathbf{v}}
\frac{S_{\mathbf{u},\mathbf{v}}}{d}
\left(A^{\mathbf{u}}_{A^2_I}\otimes A^{\mathbf{v}}_{A_O}\right)$ and $C^T = C$, and $T_{12}$ denotes the partial transpose on the first two input systems $A_I^2$. Transposing every operator on the third system turns the partial transpose on the first two systems into a full transpose, which preserves positive semidefiniteness. Therefore, the formula is equivalent to showing
\begin{equation}\label{eq:dual inequality}
    -\alpha Q_{A^2_IA_O}^{T_3}
-\beta R_{A^2_IA_O}^{T_3}
+ C_{A^2_IA_O}
\ge 0, 
\end{equation}
where $T_3$ denotes the partial transpose on the third system $A_O$. 

To verify Eq.~\eqref{eq:dual inequality}, we show the following lemma:
\begin{lemma}
    For phase space point operators $A^{\mathbf{j}}$, $\sum_{\mathbf{j}}
\left(A^{\mathbf{j}}\otimes A^{\mathbf{j}}\right) = d\mathbf{P}_2((12))$.
\end{lemma}
\begin{proof}
    For any Clifford operator $C\in \mathrm{CL}(d)$, one has the commutant relation $\left[C \otimes C, \sum_{\mathbf{j}}
\left(A^{\mathbf{j}}\otimes A^{\mathbf{j}}\right)\right] = 0$. Since the Clifford group is a unitary 2-design for any odd prime dimension~\cite{zhu2017multiqubit}, Schur's lemma gives $\sum_{\mathbf{j}}
\left(A^{\mathbf{j}}\otimes A^{\mathbf{j}}\right) = c_1I + c_2\mathbf{P}_2((12))$. Taking traces gives
\begin{equation}
        \tr\left[  \sum_{\mathbf{j}}
\left(A^{\mathbf{j}}\otimes A^{\mathbf{j}}\right)\right] = d^2, \quad \tr\left[ \mathbf{P}_2((12)) \sum_{\mathbf{j}}
\left(A^{\mathbf{j}}\otimes A^{\mathbf{j}}\right)\right] = d^3.
\end{equation}
Following~\cite[Corollary 13]{mele2024introduction}, we have $c_1 = 0$ and $c_2 = d$. This completes the proof. 
\end{proof}

Using the lemma, we obtain
\begin{equation}\label{eq:C decomposition}
\begin{aligned}
    C_{A^2_IA_O} = \sum_{\mathbf{u},\mathbf{v}}
\frac{S_{\mathbf{u},\mathbf{v}}}{d}
\left(A^{\mathbf{u}}_{A^2_I}\otimes A^{\mathbf{v}}_{A_O}\right) &= \sum_{\mathbf{u}_1,\mathbf{u}_2, \mathbf{v}} \frac{\delta(\mathbf{u}_1, \mathbf{u}_2)}{d^3p} \left(A^{\mathbf{u}_1}_{A_I}\otimes A^{\mathbf{u}_2}_{A_I}\otimes A^{\mathbf{v}}_{A_O}\right) \\
&= \frac{1}{d^3 p} \sum_{\mathbf{u}_1} A^{\mathbf{u}_1}_{A_I}\otimes A^{\mathbf{u}_1}_{A_I} \otimes \sum_{\mathbf{v}} A^{\mathbf{v}}_{A_O} \\
&= \frac{1}{dp} \mathbf{P}_3((12)).
\end{aligned}
\end{equation}
To prove $-\alpha Q_{A^2_IA_O}^{T_3}
-\beta R_{A^2_IA_O}^{T_3}
+ C_{A^2_IA_O}
\ge 0$, we apply the lemma~\ref{lem:EW-positivity}. Let $\omega = \frac{1}{t}\big(-\alpha Q_{A^2_IA_O}^{T_{3}} - \beta R_{A^2_IA_O}^{T_{3}} + C_{A^2_IA_O}\big)$ where $t = \tr \Big[-\alpha Q_{A^2_IA_O}^{T_{3}} - \beta R_{A^2_IA_O}^{T_{3}} + C_{A^2_IA_O}\Big]$. Notice $\omega$ is a Hermitian operator, and calculate
\begin{equation}
\begin{aligned}
    t &= -\alpha  - \beta d + \frac{d}{p} = \frac{(d - 1) d (d - (\delta - 2) \delta)}{\delta (d + \delta - d \delta) p} \geq 0 \\
\end{aligned}
\end{equation}
Combining Eq.~\eqref{eq: Q and R decomposition} with Eq.~\eqref{eq:C decomposition}, we calculate
\begin{equation}
    \begin{aligned}s_+ = &\tr[\omega S_+] = 1 -\frac{(d - 2) \delta}{2 d}, \quad 
   s_- = \tr[\omega S_-] = \frac{(d - 2) \delta}{2 d} \\
   s_0 = & 0,\qquad s_1=s_2=s_3=0.
    \end{aligned}
\end{equation}
We can check $s_+,s_-,s_0\geq 0$, $s_+ + s_- + s_0=1$ and $s_1^2+s_2^2+s_3^2\leq s_0^2$. Thus we show that $-\alpha Q_{A^2_IA_O}^{T_3}
-\beta R_{A^2_IA_O}^{T_3}
+ C_{A^2_IA_O}
\ge 0$. This completes the verification of dual feasibility. 

Combining the primal upper bound and the dual lower bound, we conclude that $\cM_{\cD_\delta}(2,f,p,d)  =\log \left( \cK_\cM(f - \lambda_0) + 1 \right)$, and the proof is complete.

\end{proof}


\section{Proof of Theorem~\ref{thm:robustness-universal}}
\label{app:robustness-proof}

\noindent\textbf{Theorem 2}{\rm (Robustness law for universal purification)}
\textit{For two-to-one universal probabilistic purification of multi-qubit
states under depolarizing noise, the robustness of universal purification is
bounded above and below by linear functions of the fidelity gain
\(f-\lambda_0\). For a single-qubit system, the two bounds coincide, giving an
exact linear law.}

\bigskip

\begin{proof}
More precisely, we prove that
\begin{equation}
1+
\cK_\cR^l(f-\lambda_0)
\le
\cR_{\cD_\delta}(2,f,p,d)
\le
1+
\cK_\cR^u(f-\lambda_0),
\end{equation}
where \(d=2^m\) is the dimension of an \(m\)-qubit system,
\begin{equation}
\cK_\cR^l
=
\frac{
\frac{d-2}{2}
+
\frac{1}{d-1}
+
2\delta-\delta^2
}
{\lambda_0\delta(1-\delta)},
\qquad
\cK_\cR^u
=
\frac{
2d-3+2\delta-\delta^2
}
{\lambda_0\delta(1-\delta)},
\end{equation}
and \(\lambda_0=1-\frac{d-1}{d}\delta\) is the fidelity without purification.
In particular, for a single-qubit system \(m=1\), the two bounds coincide:
\begin{equation}
\cR_{\cD_\delta}(2,f,p,2)
=
1+
\frac{1+2\delta-\delta^2}
{\lambda_0\delta(1-\delta)}
(f-\lambda_0).
\end{equation}

We first prove the upper bound
\begin{equation}
    \cR_{\cD_\delta}(2,f,p,d)
    \le
    1+\cK_\cR^u(f-\lambda_0),
\end{equation}
using the explicit two-copy purification branch and a stabilizer-state
decomposition of its Choi state.

The Choi matrix used in the proof of the mana bound, Eq.~\eqref{eq: Choi matrix parameterized}, also gives a sufficient construction:
\begin{equation}\label{eq:multi-choi-ansatz}
    J^{\cE}_{A_I^2A_O}
    =
    \mu_1\Big(\mathbf P_3((13))^{T_3}+\mathbf P_3((23))^{T_3}\Big)
    +
    \mu_2\Big(\mathbf P_3((123))^{T_3}+\mathbf P_3((132))^{T_3}\Big),
\end{equation}
where 
\begin{equation}\label{eq:multi-mu1-mu2}
    \begin{aligned}
    \mu_1 &=
    \frac{\big((d+\delta-d\delta)^2-d^2(\delta-1)^2f+d(\delta-2)\delta f\big)p}
    {2(d-1)\big(d(\delta-1)-\delta\big)(\delta-1)\delta},\\
    \mu_2 &=
    \frac{d\big(d(\delta+f-1)-\delta\big)p}
    {2(d-1)\big(d(\delta-1)-\delta\big)(\delta-1)\delta},
\end{aligned}
\end{equation}
as in Eq.~\eqref{eq: mu1 mu2}, with $d=2^m$.
As shown in the proof of Theorem~\ref{thm:mana-universal} in Appendix~\ref{app:mana-proof}, these coefficients realize the purification of prescribed success probability \(p\) and output fidelity \(f\). We now calculate the robustness of this construction.

Let
\begin{equation}
\Psi^{\cE}_{A_I^2A_O}
=
\frac{1}{d^2}J^{\cE}_{A_I^2A_O}
\end{equation}
be the corresponding Choi state. For the upper bound, we estimate \(\cR(\Psi^{\cE}_{A_I^2A_O})\). Note that
\begin{equation}
\frac{1}{d^2}\mathbf P_3((23))^{T_3}
=
\frac{I_{A_{I,1}}}{d}\otimes\Phi^+_{A_{I,2}A_O},
\qquad
\frac{1}{d^2}\mathbf P_3((13))^{T_3}
=
\Phi^+_{A_{I,1}A_O}\otimes\frac{I_{A_{I,2}}}{d},
\end{equation}
where $\Phi^+\coloneqq\ketbra{\Phi_d^+}{\Phi_d^+}$ and $\ket{\Phi_d^+}=d^{-1/2}\sum_{j=0}^{d-1}\ket{j,j}$ is the maximally entangled state, and denote
\begin{equation}
    \Omega_-
:=
\frac{
\mathbf P_3((13))^{T_3}
+
\mathbf P_3((23))^{T_3}
-
\mathbf P_3((123))^{T_3}
-
\mathbf P_3((132))^{T_3}
}{2d(d-1)}.
\end{equation}
We rewrite the Choi state as
\begin{equation}\label{eq:multi-psi-omega-decomp}
\Psi^{\cE}_{A_I^2A_O}
=
(\mu_1+\mu_2)\left(\frac{I_{A_{I,1}}}{d}\otimes\Phi^+_{A_{I,2}A_O} + \Phi^+_{A_{I,1}A_O}\otimes\frac{I_{A_{I,2}}}{d} \right)
-
\frac{2(d-1)}{d}\mu_2\,\Omega_-.
\end{equation}

Notice the first two terms $\frac{I_{A_{I,1}}}{d}\otimes\Phi^+_{A_{I,2}A_O} , \Phi^+_{A_{I,1}A_O}\otimes\frac{I_{A_{I,2}}}{d}$ are convex combinations of stabilizer states. Thus it remains to give a stabilizer decomposition of \(\Omega_-\). We will decompose \(\Omega_-\) using the stabilizer states in Eq.~\eqref{eq: chi state} and count its stabilizer $l_1$ norm.

For each \(i\in\{0,\ldots,d-1\}\), define a state
\begin{equation}
    \ket{\eta_i}
:=
\frac{1}{\sqrt{2(d-1)}}
\sum_{k\ne i}
\left(
\ket{k,i,k}_{A_I^2 A_O}-\ket{i,k,k}_{A_I^2 A_O}
\right),
\end{equation}
The \(2(d-1)\) computational basis states appearing in the above sum are
mutually orthogonal, so \(\ket{\eta_i}\) is normalized. Moreover, for
\(i\ne j\), the supports of \(\ket{\eta_i}\) and \(\ket{\eta_j}\) are
orthogonal, i.e. $\langle \eta_i|\eta_j\rangle=\delta_{ij}$. A direct expansion gives
\[
\begin{aligned}
&\mathbf P_3((13))^{T_3}
+
\mathbf P_3((23))^{T_3}
-
\mathbf P_3((123))^{T_3}
-
\mathbf P_3((132))^{T_3} \\
&\qquad =
\sum_{i=0}^{d-1}
\sum_{\substack{k\ne i\\ \ell\ne i}}
\left(
\ket{k,i,k}-\ket{i,k,k}
\right)
\left(
\bra{\ell,i,\ell}-\bra{i,\ell,\ell}
\right).
\end{aligned}
\]
Therefore, we obtain
\begin{equation}
\Omega_-
=
\frac1d
\sum_{i=0}^{d-1}
\ketbra{\eta_i}{\eta_i}.
\end{equation}

We then give an explicit stabilizer decomposition of each
\(\ketbra{\eta_i}{\eta_i}\). Define the set of indices
\begin{equation}
\mathcal S_i^{+}:=\{(k,i,k):k\ne i\},
\quad
\mathcal S_i^{-}:=\{(i,k,k):k\ne i\},
\quad
\mathcal S_i:=\mathcal S_i^{+}\cup \mathcal S_i^{-}.
\end{equation}
Define the set of unordered pairs in $\mathcal S_i$ as $\mathcal P_i$, 
\begin{equation}
\mathcal P_i
:=
\bigl\{
\{x,y\}: x,y\in\mathcal S_i,\ x\ne y
\bigr\}.
\end{equation}
where each distinct pair of elements in \(\mathcal S_i\) is counted exactly once.

Define a signed function
\begin{equation}
\epsilon_i(x)=
\begin{cases}
+1, & x\in\mathcal S_i^{+},\\
-1, & x\in\mathcal S_i^{-}.
\end{cases}
\end{equation}
With the new notation, we can rewrite $\ket{\eta_i}$ as
\begin{equation}
\ket{\eta_i}
=
\frac1{\sqrt{|\mathcal S_i|}}
\sum_{x\in\mathcal S_i}\epsilon_i(x)\ket{x},
\quad
|\mathcal S_i|=2(d-1),
\end{equation}
and
\begin{equation}\label{eq:eta-expand-full}
\ketbra{\eta_i}{\eta_i}
=
\frac1{|\mathcal S_i|}
\sum_{x\in\mathcal S_i}\ketbra{x}{x}
+
\frac1{|\mathcal S_i|}
\sum_{\{x,y\}\in\mathcal P_i}
\epsilon_i(x)\epsilon_i(y)
\left(
\ketbra{x}{y}+\ketbra{y}{x}
\right).
\end{equation}

For each unordered pair \(\{x,y\}\in\mathcal P_i\), define superposition states
\begin{equation}\label{eq: chi state}
    \ket{\chi_{x,y}^{+}}
=
\frac{\ket{x}+\epsilon_i(x)\epsilon_i(y)\ket{y}}{\sqrt2},
\qquad
\ket{\chi_{x,y}^{-}}
=
\frac{\ket{x}-\epsilon_i(x)\epsilon_i(y)\ket{y}}{\sqrt2}.
\end{equation}

Their projectors are
\begin{equation}
\ketbra{\chi_{x,y}^{+}}{\chi_{x,y}^{+}} = \frac12 \left( \ketbra{x}{x} + \ketbra{y}{y} + \epsilon_i(x)\epsilon_i(y)\ketbra{x}{y} + \epsilon_i(x)\epsilon_i(y)\ketbra{y}{x} \right),
\end{equation}
and
\begin{equation}
\ketbra{\chi_{x,y}^{-}}{\chi_{x,y}^{-}} = \frac12 \left( \ketbra{x}{x} + \ketbra{y}{y} - \epsilon_i(x)\epsilon_i(y)\ketbra{x}{y} - \epsilon_i(x)\epsilon_i(y)\ketbra{y}{x} \right).
\end{equation}
Therefore,
\begin{equation}
\begin{aligned}
&\frac1{|\mathcal S_i| - 1}\ketbra{\chi_{x,y}^{+}}{\chi_{x,y}^{+}} -\frac{|\mathcal S_i| - 2}{|\mathcal S_i| (|\mathcal S_i| - 1)}\ketbra{\chi_{x,y}^{-}}{\chi_{x,y}^{-}}\\
={}& \frac1{|\mathcal S_i| (|\mathcal S_i| - 1)} \left( \ketbra{x}{x}+\ketbra{y}{y} \right) + \frac1{|\mathcal S_i|}\epsilon_i(x)\epsilon_i(y) \left( \ketbra{x}{y}+\ketbra{y}{x} \right).
\end{aligned}
\end{equation}
For a fixed \(x\in\mathcal S_i\), the diagonal term \(\ketbra{x}{x}\) appears
in exactly \(|\mathcal S_i| - 1\) unordered pairs. Hence its total coefficient is $(|\mathcal S_i| - 1)\cdot \frac1{|\mathcal S_i| (|\mathcal S_i| - 1)}=\frac1{|\mathcal S_i|}$.
Thus
\begin{equation}\label{eq:Omega-minus-eta-mixture}
  \Omega_-
=
\frac1d
\sum_{i=0}^{d-1}
\ketbra{\eta_i}{\eta_i} =
\frac1d
\sum_{i=0}^{d-1}\sum_{\{x,y\}\in\mathcal P_i}
\left(
\frac1{|\mathcal S_i| - 1}\ketbra{\chi_{x,y}^{+}}{\chi_{x,y}^{+}}
-
\frac{|\mathcal S_i| - 2}{|\mathcal S_i| (|\mathcal S_i| - 1)}\ketbra{\chi_{x,y}^{-}}{\chi_{x,y}^{-}}
\right).
\end{equation}

We now explain why these are stabilizer states. Every unordered pair \(\{x,y\}\in\mathcal P_i\) is of one of the following three types.
\begin{enumerate}
    \item If $x=(k,i,k)\in\mathcal S_i^+,
\;
y=(\ell,i,\ell)\in\mathcal S_i^+,
\;
k,\ell\ne i,\; k\ne \ell,$
then \(\epsilon_i(x)\epsilon_i(y)=+1\), and the two states are
\begin{equation}
\ket{\chi_{x,y}^{\pm}} = \frac{\ket{k,i,k}\pm\ket{\ell,i,\ell}}{\sqrt2} = \frac{\ket{k,k}_{A_{I,1} A_O}\pm\ket{\ell,\ell}_{A_{I,1} A_O}}{\sqrt2} \otimes \ket{i}_{A_{I,2}}
\end{equation}
The second register \(A_{I,2}\) is fixed to \(\ket{i}\), while the first and
third registers form a two-point Bell-type stabilizer state. Hence these are stabilizer states, Clifford equivalent to a one-qubit Bell state tensored with computational basis states.

\item If $x=(i,k,k)\in\mathcal S_i^-,
\;
y=(i,\ell,\ell)\in\mathcal S_i^-,
\;
k,\ell\ne i,\; k\ne \ell,$
then again \(\epsilon_i(x)\epsilon_i(y)=+1\). The same argument as in the first case shows that these are stabilizer states.

\item If $x=(k,i,k)\in\mathcal S_i^+,
\;
y=(i,\ell,\ell)\in\mathcal S_i^-,$
then \(\epsilon_i(x)\epsilon_i(y)=-1\). Hence
\begin{equation}
\ket{\chi_{x,y}^{\pm}}
=
\frac{\ket{x}\mp\ket{y}}{\sqrt2}.
\end{equation}
It is less clear to observe that they are stabilizer states. 
Regard \(x\) and \(y\) as binary strings in \(\mathbb F_2^{3m}\). Apply the
Pauli \(X\)-operator corresponding to the bit string \(x\), 
$X^x:=\bigotimes_{r=1}^{3m}X_r^{x_r}$. This maps computational basis states as $X^x\ket{z}=\ket{z\oplus x}$.
Therefore
\begin{equation}
    X^x\ket{x}=\ket{0}^{\otimes 3m},
\qquad
X^x\ket{y}=\ket{x\oplus y},
\end{equation}
where $x\oplus y$ denotes addition modulo 2. 
Since \(x\ne y\), we have \(x\oplus y\ne0\). Now choose an invertible linear map $L:\mathbb F_2^{3m}\to\mathbb F_2^{3m}$
such that
\begin{equation}
    L(x\oplus y)=(1,0,\ldots,0).
\end{equation}
Such a map exists because any nonzero vector can be extended to a basis of
\(\mathbb F_2^{3m}\). Moreover, any invertible linear map over \(\mathbb F_2\)
can be implemented by CNOT gates and qubit permutations~\cite{markov2008optimal}, hence by a Clifford
unitary. Let \(U_L\) denote the corresponding Clifford unitary, so that $U_L\ket{z}=\ket{Lz}.$ Since \(L\) is linear, \(L(0)=0\). Hence
\begin{equation}
U_LX^x\ket{x}=\ket{0}^{\otimes 3m},
\qquad
U_LX^x\ket{y}=\ket{1} \otimes \ket{0}^{\otimes (3m-1)}.
\end{equation}
Thus,
\begin{equation}
U_LX^x
\frac{\ket{x}\mp\ket{y}}{\sqrt2}
=
\frac{\ket{0}^{\otimes 3m}\mp\ket{1} \otimes \ket{0}^{\otimes (3m-1)}}{\sqrt2}
=
\ket{\mp}\otimes\ket{0}^{\otimes(3m-1)},
\end{equation}
Since Clifford unitaries preserve stabilizer states, we conclude that both
\(\ket{\chi_{x,y}^{+}}\) and \(\ket{\chi_{x,y}^{-}}\) are stabilizer states.
\end{enumerate}

The discussion above shows that all \(\ket{\chi_{x,y}^{+}}\) and \(\ket{\chi_{x,y}^{-}}\) are stabilizer states. Thus Eq.~\eqref{eq:Omega-minus-eta-mixture} is a stabilizer decomposition. We next count the corresponding \(l_1\)-norm and then bound the robustness of the Choi state. For each fixed \(\ketbra{\eta_i}{\eta_i}\) in Eq.~\eqref{eq:eta-expand-full}, there are $\binom{|\mathcal S_i|} {2}=\frac{|\mathcal S_i| (|\mathcal S_i| - 1)}2$ unordered pairs, whose stabilizer decomposition has total \(l_1\)-norm
\[
\binom{|\mathcal S_i|} {2} \cdot
\left(
\left|\frac1{|\mathcal S_i| - 1}\right|
+
\left|-\frac{|\mathcal S_i| - 2}{|\mathcal S_i| (|\mathcal S_i| - 1)}\right|
\right)
=
\frac{|\mathcal S_i| (|\mathcal S_i| - 1)}2\cdot\frac{2}{|\mathcal S_i|}
=
|\mathcal S_i| - 1=2d-3.
\]

Overall, we write $\Psi^{\cE}_{A_I^2A_O}$ in terms of stabilizer decomposition,
\begin{equation}
\begin{aligned}
        \Psi^{\cE}_{A_I^2A_O} &= (\mu_1+\mu_2)\left(\frac{I_{A_{I,1}}}{d}\otimes\Phi^+_{A_{I,2}A_O} + \Phi^+_{A_{I,1}A_O}\otimes\frac{I_{A_{I,2}}}{d} \right) \\
&-
\frac{2(d-1)}{d}\mu_2 \cdot \frac1d
\sum_{i=0}^{d-1}\sum_{\{x,y\}\in\mathcal P_i}
\left(
\frac1{|\mathcal S_i| - 1}\ketbra{\chi_{x,y}^{+}}{\chi_{x,y}^{+}}
-
\frac{|\mathcal S_i| - 2}{|\mathcal S_i| (|\mathcal S_i| - 1)}\ketbra{\chi_{x,y}^{-}}{\chi_{x,y}^{-}}
\right).
\end{aligned}
\end{equation}

Since $\mu_1 \geq 0$ and $\mu_2 \geq 0$, this decomposition gives
\begin{equation}
\begin{aligned}
\cR(\Psi^{\cE}_{A_I^2A_O})
&\le
2(\mu_1+\mu_2)
+
\frac{2(d-1)}{d}\mu_2\cdot \frac1d \cdot d \cdot (2d-3) \\
& =
2(\mu_1+\mu_2)
+
\frac{2(d-1)(2d-3)}{d}\mu_2  \\
&\le p\left[ 1+
\frac{2d-3+2\delta-\delta^2}
{\lambda_0\delta(1-\delta)}
(f-\lambda_0) \right] \\
&= p\left[ 1+
\cK_\cR^u
(f-\lambda_0) \right].
\end{aligned}
\end{equation}
Here $\cK_\cR^u = \frac{2d-3+2\delta-\delta^2}
{\lambda_0\delta(1-\delta)}$.
Thus
\begin{equation}
    \cR_{\cD_\delta}(2,f,p,d)
    \le \frac{1}{p}\cR(\Psi^{\cE}_{A_I^2A_O})
    \le 1+
\cK_\cR^u
(f-\lambda_0).
\end{equation}
This completes the upper-bound calculation for the robustness of purification.

Next, we prove the lower bound
\begin{equation}
\cR_{\cD_\delta}(2,f,p,d)
\ge
1+
\cK_\cR^l
(f-\lambda_0),
\end{equation}
where $\cK_\cR^l = \frac{
\frac{d-2}{2}
+
\frac{1}{d-1} + 2\delta-\delta^2
} {\lambda_0\delta(1-\delta)}$. 

The lower bound is proved by reducing an arbitrary feasible protocol to a symmetric Choi state and then evaluating a robustness witness. The first step is justified by Clifford twirling.
\begin{lemma}
For any CPTN map \(\mathcal{E}_{A_I^n\to A_O}\) satisfying
\begin{equation}
\tr\left[
J^{\mathcal{E}}_{A_I^nA_O}
Q^{T_{A_I^n}}_{A_I^nA_O}
\right]
= pf,
\qquad
\tr\left[
J^{\mathcal{E}}_{A_I^nA_O}
R^{T_{A_I^n}}_{A_I^nA_O}
\right]
= p,
\end{equation}
let \(\Psi^\mathcal{E}_{A_I^nA_O}=J^\mathcal{E}_{A_I^nA_O}/d_{A_I^n}\) denote its Choi state. Then the stabilizer robustness of the Choi state of its Clifford twirl is no larger than that of the original map:
\begin{equation}
        \cR(\Psi_{A^n_IA_O}^{\mathcal{T}_{\mathrm{Cl}}(\mathcal{E})})
\le
\cR(\Psi^\mathcal{E}_{A_I^nA_O}),
\end{equation}  
where
\begin{equation}
\mathcal{T}_{\mathrm{Cl}}(\mathcal{E})
=
\frac{1}{|\mathrm{Cl}(d)|}
\sum_{C\in \mathrm{Cl}(d)}
\cC(\cE),
\qquad
\cC(\cE)
=
\operatorname{Ad}_{C}\circ \mathcal{E}\circ \operatorname{Ad}_{C^\dagger}^{\otimes n}.
\end{equation}
\end{lemma}

\begin{proof}
Since each \(\cC(\cE)\) is CPTN whenever \(\mathcal{E}\) is CPTN, and
\(\mathcal{T}_{\mathrm{Cl}}(\mathcal{E})\) is a convex combination of such maps,
\(\mathcal{T}_{\mathrm{Cl}}(\mathcal{E})\) is also CPTN. We first show that the Clifford twirl preserves the success probability \(p\) and the
quantity \(pf\). For the success probability,
\begin{equation}
\begin{aligned}
\int d\psi 
\tr\left[
\mathcal{T}_{\mathrm{Cl}}(\mathcal{E})(\cD_\delta(\psi)^{\otimes n})
\right]
&=
\int d\psi 
\frac{1}{|\mathrm{Cl}(d)|}
\sum_{C\in \mathrm{Cl}(d)}
\tr\left[
\cC(\cE)(\cD_\delta(\psi)^{\otimes n})
\right] \\
&=
\frac{1}{|\mathrm{Cl}(d)|}
\sum_{C\in \mathrm{Cl}(d)}
\int d\psi 
\tr\left[
C \mathcal{E} \left((C^\dagger \cD_\delta(\psi) C)^{\otimes n}\right) C^\dagger
\right] \\
&=
\frac{1}{|\mathrm{Cl}(d)|}
\sum_{C\in \mathrm{Cl}(d)}
\int d\psi 
\tr\left[
\mathcal{E} \left(\cD_\delta(C^\dagger \psi C)^{\otimes n}\right)
\right] \\
&=
\frac{1}{|\mathrm{Cl}(d)|}
\sum_{C\in \mathrm{Cl}(d)}
\int d\phi 
\tr\left[
\mathcal{E}(\cD_\delta(\phi)^{\otimes n})
\right] \\
&=
\int d\phi 
\tr\left[
\mathcal{E}(\cD_\delta(\phi)^{\otimes n})
\right] \\
&=
p.
\end{aligned}
\end{equation}
Here we use the unitary invariance of the depolarizing channel \(\cD_\delta\) and the invariance of the
Haar measure under the change of variables \(\phi = C^\dagger \psi C\). Similarly, for the quantity \(pf\),
\begin{equation}
\begin{aligned}
\int d\psi 
\tr\left[
\psi \mathcal{T}_{\mathrm{Cl}}(\mathcal{E})(\cD_\delta(\psi)^{\otimes n})
\right]
&=
\frac{1}{|\mathrm{Cl}(d)|}
\sum_{C\in \mathrm{Cl}(d)}
\int d\psi 
\tr\left[
\psi 
C \mathcal{E} \left(\cD_\delta(C^\dagger \psi C)^{\otimes n}\right)C^\dagger
\right] \\
&=
\frac{1}{|\mathrm{Cl}(d)|}
\sum_{C\in \mathrm{Cl}(d)}
\int d\psi 
\tr\left[
C^\dagger \psi C 
\mathcal{E} \left(\cD_\delta(C^\dagger \psi C)^{\otimes n}\right)
\right] \\
&=
\frac{1}{|\mathrm{Cl}(d)|}
\sum_{C\in \mathrm{Cl}(d)}
\int d\phi 
\tr\left[
\phi \mathcal{E}(\cD_\delta(\phi)^{\otimes n})
\right] \\
&=
\int d\phi 
\tr\left[
\phi \mathcal{E}(\cD_\delta(\phi)^{\otimes n})
\right] \\
&=
pf.
\end{aligned}
\end{equation}
Therefore, \(\mathcal{T}_{\mathrm{Cl}}(\mathcal{E})\) preserves both constraints.

It remains to show that Clifford twirling does not increase the robustness of
the Choi state. For each Clifford operator \(C\), the Choi state of
the corresponding twirled component \(\cC(\cE)\) is
\begin{equation}
\Psi_{A^n_IA_O}^{\cC(\cE)}
=
\bigl((C^*)^{\otimes n}\otimes C\bigr)
\Psi_{A^n_IA_O}^\mathcal{E}
\bigl((C^*)^{\otimes n}\otimes C\bigr)^\dagger.
\end{equation}
Since Clifford unitaries map stabilizer states to stabilizer states, the
robustness is invariant under this conjugation:
\begin{equation}
\cR(\Psi_{A^n_IA_O}^{\cC(\cE)})
=
\cR(\Psi_{A^n_IA_O}^\mathcal{E}).
\end{equation}
Therefore, by the convexity of the robustness of magic,
\begin{equation}
\begin{aligned}
\cR(\Psi_{A^n_IA_O}^{\mathcal{T}_{\mathrm{Cl}}(\mathcal{E})})
=
\cR\!\left(
\frac1{|\mathrm{Cl}(d)|}
\sum_{C\in\mathrm{Cl}(d)}
\Psi_{A^n_IA_O}^{\cC(\cE)}
\right)
\le
\frac1{|\mathrm{Cl}(d)|}
\sum_{C\in\mathrm{Cl}(d)}
\cR(\Psi_{A^n_IA_O}^{\cC(\cE)})
=
\cR(\Psi_{A^n_IA_O}^\mathcal{E}).
\end{aligned}
\end{equation}
This proves the lemma.
\end{proof} 

We now set \(n=2\) and write \(\Psi^\cE_{A_I^2A_O}=J^\cE_{A_I^2A_O}/d^2\) for the Choi state of a feasible two-copy map.
Since the multi-qubit Clifford group is a unitary 3-design, the Clifford-twirled Choi state lies in the commutant of \(U\otimes U\otimes \overline U\), equivalently in the partially transposed permutation algebra. We use the basis $\{S_i\}$ from Lemma~\ref{lem:EW-positivity} and write
\begin{equation}
s_i = \tr[\Psi^\cE_{A_I^2A_O} S_i],\qquad i\in \{+, -, 0,1,2,3\}.
\end{equation}

In addition, we may symmetrize over the input-swap operation exchanging \(A_{I,1}\) and \(A_{I,2}\). This preserves the success probability and fidelity constraints, because \(Q\) and \(R\) are input-symmetric, and it does not increase robustness by convexity. Therefore, without loss of generality, the Choi state satisfies
\begin{equation}
\bP_3((12))\Psi^\cE_{A_I^2A_O}\bP_3((12))
=
\Psi^\cE_{A_I^2A_O}.
\end{equation}
We note that
\begin{equation}
\bP_3((12)) S_j \bP_3((12))=S_j,\qquad j\in\{+,-,0,1\},
\end{equation}
whereas
\begin{equation}
\bP_3((12)) S_2 \bP_3((12))=-S_2,
\qquad
\bP_3((12)) S_3 \bP_3((12))=-S_3.
\end{equation}
Thus $s_2 = s_3 = 0$, and we only need to consider $\{s_+, s_-, s_0, s_1\}$. The positivity criterion gives 
\begin{equation}\label{eq:S-positivity-criterion}
    s_+\ge0,\quad s_-\ge0,\quad s_0\ge0,\quad
s_1^2\le s_0^2.
\end{equation}
The number of free parameters can be further reduced to $\{s_+, s_-\}$ by solving the purification constraints
\begin{equation}
\tr[\Psi^\cE_{A_I^2A_O} Q^{T_{A_I^2}}_{A_I^2A_O}]= \frac{pf}{d^2},
\qquad
\tr[\Psi^\cE_{A_I^2A_O} R^{T_{A_I^2}}_{A_I^2A_O}]=\frac{p}{d^2}.
\end{equation}
For clarity and simplicity, we denote $Q_{A_I^2A_O}, R_{A_I^2A_O}$ calculated in Eq.~\eqref{eq: Q and R decomposition} as
\begin{equation}
\begin{aligned}
    Q_{A_I^2A_O}
&=
q_1 I+q_2 \bP_3((12))+q_3(\bP_3((23))+\bP_3((13)))+q_4(\bP_3((132))+\bP_3((123))),
\\
R_{A_I^2A_O}
&=
r_1 I+r_2\bP_3((12)),
\end{aligned}
\end{equation}
with
\begin{equation}
    q_1 =
\frac{d^2+4d\delta +(2-d)\delta^2}{(d+2)(d+1)d^3},\quad
q_3 =
\frac{(1-\delta)(d+2\delta)}{(d+2)(d+1)d^2},
\quad
q_2 = q_4 =
\frac{(1-\delta)^2}{(d+2)(d+1)d},
\end{equation}
and
\begin{equation}
    r_1=
\frac{d+\delta(2-\delta)}{(d+1)d^2},
\qquad
r_2=
\frac{(1-\delta)^2}{(d+1)d}.
\end{equation}

The equations become
\begin{equation}\label{eq:Q-S-constraint}
(q_1+q_2)s_+
+
(q_1-q_2)s_-
+
(q_1+dq_3+q_4)s_0
+
(q_2+q_3+dq_4)s_1
=
\frac{pf}{d^2},
\end{equation}
and
\begin{equation}\label{eq:R-S-constraint}
(r_1+r_2)s_+
+
(r_1-r_2)s_-
+
r_1s_0
+
r_2s_1
=
\frac{p}{d^2}.
\end{equation}

Now solve Eqs.~\eqref{eq:Q-S-constraint} and \eqref{eq:R-S-constraint} for
\(s_0,s_1\). Let
\begin{equation}
\Delta = (q_1+dq_3+q_4)r_2-(q_2+q_3+dq_4)r_1 = -\frac{\delta(1-\delta)\bigl(d(1-\delta)+\delta\bigr)}
{d^4(d+1)}
\neq0
\end{equation}
for \(d\ge2\) and \(0<\delta<1\). Hence the solution is
\begin{equation}\label{eq:s0-solution}
s_0
=
\frac{
\left(\frac{pf}{d^2}-(q_1+q_2)s_+-(q_1-q_2)s_-\right)r_2
-
\left(\frac{p}{d^2}-(r_1+r_2)s_+-(r_1-r_2)s_-\right)(q_2+q_3+dq_4)
}{\Delta},
\end{equation}
and
\begin{equation}\label{eq:s1-solution}
s_1
=
\frac{
(q_1+dq_3+q_4)\left(\frac{p}{d^2}-(r_1+r_2)s_+-(r_1-r_2)s_-\right)
-
r_1\left(\frac{pf}{d^2}-(q_1+q_2)s_+-(q_1-q_2)s_-\right)
}{\Delta}.
\end{equation}
Thus the purification constraints reduce the remaining free parameters to $\{s_+, s_-\}$. 

We now choose a witness whose expectation can be expressed in these remaining parameters and whose value is bounded on all stabilizer projectors.
Consider the two operators
\begin{equation}
W_1
=
\frac{-2I+(d-1)\bP_3((12))+\bP_3((23))^{T_3}+\bP_3((13))^{T_3}+\bP_3((132))^{T_3}+\bP_3((123))^{T_3}}{d+1}
\end{equation}
and
\begin{equation}
    W_2 = \frac{1}{2}(\bP_3((132))^{T_3}+\bP_3((123))^{T_3}).
\end{equation}

Define
\begin{equation}\label{eq:Wsm-def}
\begin{aligned}
W &= \frac1{d-1}W_1 + \left(1 - \frac1{d-1}\right)W_2 \\
&= 
\frac{1}{(d-1)(d+1)}  \left(-2I + 
\bP_3((23))^{T_3}
+
\bP_3((13))^{T_3}
\right)
+
\frac{1}{d+1}\bP_3((12)) 
 \\
&+
\frac{d}{2(d+1)}
\left(
\bP_3((123))^{T_3}
+
\bP_3((132))^{T_3}
\right).
\end{aligned}
\end{equation}
We will prove the robustness lower bound through the inequalities
\begin{equation}
    \cR(\Psi^\cE_{A_I^2A_O})
    \geq
    \tr[W\Psi^\cE_{A_I^2A_O}]
    \geq
    p\left[
    1+\cK_\cR^l(f-\lambda_0)
    \right].
\end{equation}

We prove the following lemma:
\begin{lemma}\label{lem:MD-dual-feasible}
For every pure tripartite \(3m\)-qubit stabilizer projector \(\phi\), one has \(|\tr[W_1\phi]|\le 1\).
\end{lemma}

\begin{proof}
Let \(\phi=\ketbra{\psi}{\psi}\) be a pure tripartite stabilizer projector. Define
\begin{equation}
x_{12}=\tr[\bP_3((12))\phi],
\quad
x_{23}=\tr[\bP_3((23))^{T_3}\phi],
\quad
x_{13}=\tr[\bP_3((13))^{T_3}\,\phi],
\quad
x_{123}=\tr[\bP_3((132))^{T_3}\phi].
\end{equation}
We note that
\begin{equation}
\tr[(\bP_3((132))^{T_3}+\bP_3((123))^{T_3})\phi]
=
x_{123}+\overline{x_{123}}
=
2\operatorname{Re}(x_{123}).
\end{equation}
Therefore
\begin{equation}\label{eq:MD-expectation}
\tr[W_1\phi]
=
\frac{
-2+(d-1)x_{12}+x_{23}+x_{13}+2\operatorname{Re} (x_{123})
}{d+1}.
\end{equation}
We will first prove that $\tr[W_1\phi] \le 1$, which is to show
\begin{equation}
-2+(d-1)x_{12}+x_{23}+x_{13}+2\operatorname{Re} (x_{123})
\le d+1.
\end{equation}
To show this, we use the tripartite stabilizer normal form to identify the block structure of the state. We use the following lemma.
\begin{lemma}[Tripartition of stabilizer states~\cite{bravyi2006ghz, looi2011tripartite}]
For any tripartite qubit stabilizer state $\ket{\phi} \in \cH_A \otimes \cH_B \otimes \cH_C$, there exist local Clifford unitaries $U_A, U_B, U_C$ such that $U_A \otimes U_B \otimes U_C \ket{\phi}$ is a tensor product of GHZ states, maximally entangled states, and local single-qubit stabilizer states.
\end{lemma}

In our case, we express the local Clifford decomposition for the multi-qubit system as follows:
\begin{equation}
    \phi_{{A_{I}^2}{A_{O}}} = (U_{A_{I,1}} \otimes U_{A_{I,2}} \otimes U_{A_O}) \phi^{\mathrm{nf}}_{A_{I}^2  A_O}  (U_{A_{I,1}} \otimes U_{A_{I,2}} \otimes U_{A_O})^\dagger,  
\end{equation}
where \(\phi^{\mathrm{nf}}_{A_I^2A_O}\) denotes the tripartite stabilizer normal form, written as a tensor product of elementary one-qubit blocks. These blocks are GHZ states, maximally entangled states, or local single-qubit stabilizer states. Let
\begin{itemize}
    \item $m_{\mathrm{GHZ}}$ be the number of GHZ states;
    \item $m_{12}$ be the number of Bell pairs between $A_{I,1}$ and $A_{I,2}$, tensored with a single-qubit state on $A_O$;
    \item $m_{13}$ be the number of Bell pairs between $A_{I,1}$ and $A_O$, tensored with a single-qubit state on $A_{I,2}$;
    \item $m_{23}$ be the number of Bell pairs between $A_{I,2}$ and $A_O$, tensored with a single-qubit state on $A_{I,1}$;
    \item $m_{\mathrm{loc}}$ be the number of local single-qubit stabilizer states.
\end{itemize}
These numbers satisfy $m_{\mathrm{GHZ}} + m_{12} + m_{13} + m_{23} + m_{\mathrm{loc}} = m$.
We can write $\phi^{\mathrm{nf}}_{A_{I}^2  A_O}$ as
\begin{equation}
\begin{aligned}
    \phi^{\mathrm{nf}}_{A_{I}^2  A_O} &= \psi_{\mathrm{GHZ}}^{\otimes m_{\mathrm{GHZ}}} \otimes (\Phi^+_{A_{I}^2} \otimes \ketbra{+}{+}_{A_O})^{\otimes m_{12}} \otimes (\Phi^+_{A_{I,1} A_{O}} \otimes \ketbra{+}{+}_{A_{I,2}})^{\otimes m_{13}} \\
    & \otimes (\Phi^+_{A_{I,2} A_{O}} \otimes \ketbra{+}{+}_{A_{I,1}})^{\otimes m_{23}} \otimes \ketbra{+++}{+++}^{\otimes m_{\mathrm{loc}}},
\end{aligned}
\end{equation}
where \(\ket{+}\), \(\Phi^+\), and the standard GHZ state are chosen only as canonical representatives of the elementary stabilizer blocks. Any other one-qubit stabilizer state, Bell state, or GHZ stabilizer representative is Clifford equivalent to these choices and can be absorbed into the Clifford unitaries. The estimates below are taken over the resulting local Clifford orbit, so they do not depend on this particular choice of representatives. In addition, we also use 
\begin{equation}\label{eq:partial trace A1 of phi}
\begin{aligned}
    \tr_{A_{I, 1}}\phi^{\mathrm{nf}}_{A_{I}^2  A_O} &= \rho_{\mathrm{GHZ}}^{\otimes m_{\mathrm{GHZ}}} \otimes (\frac{I_{A_{I,2}}}{2} \otimes \ketbra{+}{+}_{A_O})^{\otimes m_{12}} \otimes ( \ketbra{+}{+}_{A_{I,2}} \otimes \frac{I_{A_{O}}}{2})^{\otimes m_{13}} \\
    & \otimes (\Phi^+_{A_{I,2} A_{O}} )^{\otimes m_{23}} \otimes \ketbra{++}{++}^{\otimes m_{\mathrm{loc}}},
\end{aligned}
\end{equation}
where $\rho_{\mathrm{GHZ}} = \frac{1}{2}(\ketbra{00}{00} + \ketbra{11}{11})$. Partial traces over the other two systems have analogous forms.

We will also use
\begin{equation}
\bP_3((23))^{T_3}
=
d I_{A_{I,1}}\otimes \Phi^+_{A_{I,2}A_O},
\end{equation}
We first calculate $x_{23}$:
\begin{equation}
\begin{aligned}
    x_{23} &= \tr[\bP_3((23))^{T_3} \phi_{{A_{I}^2}{A_{O}}}] \\
    &= d \tr\left[I_{A_{I,1}} \otimes \Phi^+_{A_{I,2}, A_O} \phi_{{A_{I}^2}{A_{O}}}\right] \\
    &= d \tr\left[ \Phi^+_{A_{I,2}, A_O} \tr_{A_{I,1}}\phi_{{A_{I}^2}{A_{O}}}\right] \\
    &= d \tr\left[ \Phi^+_{A_{I,2}, A_O}  (U_{A_{I,2}} \otimes U_{A_O}) \tr_{A_{I,1}} \phi^{\mathrm{nf}}_{A_{I}^2 A_O} (U_{A_{I,2}} \otimes U_{A_O})^\dagger \right] \\
    &\le \max_{U_{A_{I,2}} , U_{A_O}} d \tr\left[ \Phi^+_{A_{I,2}, A_O}  (U_{A_{I,2}} \otimes U_{A_O}) \tr_{A_{I,1}} \phi^{\mathrm{nf}}_{A_{I}^2 A_O} (U_{A_{I,2}} \otimes U_{A_O})^\dagger \right] \\
    &\le \max_{\widetilde{\Phi}^{+}_{A_{I,2} A_O}} d \tr\left[ \widetilde{\Phi}^{+}_{A_{I,2} A_O}   \tr_{A_{I,1}} \phi^{\mathrm{nf}}_{A_{I}^2 A_O}  \right],
\end{aligned}
\end{equation}
where $\widetilde{\Phi}^{+}_{A_{I,2} A_O} = (U_{A_{I,2}} \otimes U_{A_O})^\dagger\Phi^+_{A_{I,2}, A_O}(U_{A_{I,2}} \otimes U_{A_O})$ is still a maximally entangled state.
We bound $\tr_{A_{I, 1}}\phi^{\mathrm{nf}}_{A_{I}^2  A_O}$ by the inequalities
\begin{equation}\label{eq:approx ghz and epr}
      \rho_{\mathrm{GHZ}} \leq \frac{I_{A_{I,2} A_O}}{2}, \quad \Phi^+_{A_{I,2} A_O}  \leq I_{A_{I,2} A_O}.
\end{equation}
Thus
\begin{equation}\label{eq: bound of partial-trace phi^nf}
\begin{aligned}
    \tr_{A_{I, 1}}\phi^{\mathrm{nf}}_{A_{I}^2  A_O} &\le (\frac{I_{A_{I,2} A_O}}{2})^{\otimes m_{\mathrm{GHZ}}} \otimes (\frac{I_{A_{I,2}}}{2} \otimes \ketbra{+}{+}_{A_O})^{\otimes m_{12}} \otimes ( \ketbra{+}{+}_{A_{I,2}} \otimes \frac{I_{A_{O}}}{2})^{\otimes m_{13}} \\
    & \otimes I_{A_{I,2}, A_O} ^{\otimes m_{23}} \otimes \ketbra{++}{++}^{\otimes m_{\mathrm{loc}}} \\
    &= \frac{ I_{A_{I,2}}^{\otimes (m_{\mathrm{GHZ}} + m_{12}  + m_{23})} \otimes \ketbra{+}{+}_{A_{I,2}} ^ {\otimes (m_\mathrm{loc}+ m_{13}) }}{2^{m_{\mathrm{GHZ}} + m_{12}  + m_{13} }} \\
    &\otimes  I_{A_{O}}^{\otimes (m_{\mathrm{GHZ}} + m_{13} + m_{23})} \otimes \ketbra{+}{+}_{A_{O}} ^ {\otimes ( m_{12}  + m_\mathrm{loc})} \\
    &= \frac{E_{A_{I,2}} \otimes F_{A_O}}{2^{m_{\mathrm{GHZ}} + m_{12} + m_{13} }}. 
\end{aligned}
\end{equation}
Thus 
\begin{equation}
\begin{aligned}
    x_{23}
    &\le \max_{\widetilde{\Phi}^{+}_{A_{I,2} A_O}} 
    d \tr\left[ \widetilde{\Phi}^{+}_{A_{I,2} A_O}   
    \tr_{A_{I,1}} \phi^{\rm nf}_{A_{I,1} A_{I,2} A_O}  \right] \\
    &\le \max_{\widetilde{\Phi}^{+}_{A_{I,2} A_O}}
    \frac{d}{2^{m_{\mathrm{GHZ}} + m_{12} + m_{13} }}
    \tr\left[ \widetilde{\Phi}^{+}_{A_{I,2} A_O}   
    (E_{A_{I,2}} \otimes F_{A_O})  \right] \\
    &\le
    \frac{\min \{\mathrm{rank}\, E_{A_{I,2}},  
    \mathrm{rank}\, F_{A_O}\}}
    {2^{m_{\mathrm{GHZ}} + m_{12} + m_{13} }} \\
    &=
    \frac{
    2^{m_{\mathrm{GHZ}} + m_{23}
    +\min \{m_{12}, m_{13}\}}
    }
    {2^{m_{\mathrm{GHZ}} + m_{12} + m_{13}}} \\
    &=
    2^{m_{23}
    -\max \{m_{12}, m_{13}\}} \\
\end{aligned}
\end{equation}
where we use $\tr[\Phi^{+} (E \otimes F)] = \frac{1}{d}\tr[E^TF] \leq \frac{1}{d}\min \{\mathrm{rank} \;E,  \mathrm{rank} \;F\}$ for maximally entangled states. By the same argument, we obtain $x_{13} \leq 2^{m_{13}
    -\max \{m_{12}, m_{23}\}}$ by exchanging the input systems. 
    
For $x_{12}$, the calculation is slightly different: 
\begin{equation}
\begin{aligned}
    x_{12} &= \tr[\bP_3((12)) \phi_{{A_{I}^2}{A_{O}}}] \\
    &=  \tr\left[ \bP_2((12)) \otimes I_{A_O} \phi_{{A_{I}^2}{A_{O}}}\right] \\
    &=  \tr\left[ \bP_2((12)) \tr_{A_{O}}\phi_{{A_{I}^2}{A_{O}}}\right] \\
    &=  \tr\left[ \bP_2((12))  (U_{A_{I,1}} \otimes U_{A_{I,2}}) \tr_{A_{O}} \phi^{\mathrm{nf}}_{A_{I}^2 A_O} (U_{A_{I,1}} \otimes U_{A_{I,2}})^\dagger \right] \\
    &\le \max_{U_{A_{I,1}} , U_{A_{I,2}}}  \tr\left[ \bP_2((12))  (U_{A_{I,1}} \otimes U_{A_{I,2}}) \tr_{A_{O}} \phi^{\mathrm{nf}}_{A_{I}^2 A_O} (U_{A_{I,1}} \otimes U_{A_{I,2}})^\dagger \right]. \\
\end{aligned}
\end{equation}
Applying the same argument as in Eq.~\eqref{eq: bound of partial-trace phi^nf} with the partial trace over $A_O$, we write 
\begin{equation}
    \tr_{A_{O}} \phi^{\mathrm{nf}}_{A_{I}^2 A_O} \leq \frac{E_{A_{I,1}} \otimes F_{A_{I,2}}}{2^{m_{\mathrm{GHZ}} + m_{23} + m_{13} }}.
\end{equation}
Then we have
\begin{equation}
    \begin{aligned}
        x_{12}
    & \le \max_{U_{A_{I,1}} , U_{A_{I,2}}}  \tr\left[ \bP_2((12))  (U_{A_{I,1}} \otimes U_{A_{I,2}}) \tr_{A_{O}} \phi^{\mathrm{nf}}_{A_{I}^2 A_O} (U_{A_{I,1}} \otimes U_{A_{I,2}})^\dagger \right] \\
    &\le \max_{U_{A_{I,1}} , U_{A_{I,2}}}
    \frac{1}{2^{m_{\mathrm{GHZ}} + m_{23} + m_{13} }}
    \tr\left[ \bP_2((12))   
    (U_{A_{I,1}}E_{A_{I,1}}U_{A_{I,1}}^\dagger \otimes U_{A_{I,2}}F_{A_{I,2}}U_{A_{I,2}}^\dagger)  \right] \\
    &\le
    \frac{\min \{\mathrm{rank}\, E_{A_{I,1}},  
    \mathrm{rank}\, F_{A_{I,2}}\}}
    {2^{m_{\mathrm{GHZ}} + m_{23} + m_{13} }} \\
    &=
    \frac{
    2^{m_{\mathrm{GHZ}} + m_{12}
    +\min \{m_{13}, m_{23}\}}
    }
    {2^{m_{\mathrm{GHZ}} + m_{23} + m_{13}}} \\
    &=
    2^{m_{12}
    -\max \{m_{13}, m_{23}\}}, 
    \end{aligned}
\end{equation}
where we use the swap trick $\tr[\bP_2((12)) (E \otimes F)] = \tr[EF] \leq \min \{\mathrm{rank} \;E,  \mathrm{rank} \;F\}$.
Since \(\bP_3((12))\) is a Hermitian unitary, we also have \(x_{12}\le 1\). Hence
\[
x_{12}
\le
\min\left\{
1,\;
2^{m_{12}-\max\{m_{13},m_{23}\}}
\right\}.
\]

Summarizing the preceding calculation, we have
\begin{equation}
x_{23}
\le
2^{m_{23}-\max\{m_{12},m_{13}\}},
\quad
x_{13}
\le
2^{m_{13}-\max\{m_{12},m_{23}\}},
\quad
x_{12}
\le
\min\left\{
1,\,
2^{m_{12}-\max\{m_{13},m_{23}\}}
\right\}.
\end{equation}
Moreover, since \(\bP_3((23))^{T_3}\ge0\) and
\(\bP_3((132))^{T_3}=\bP_3((23))^{T_3}\bP_3((12))\), the Cauchy--Schwarz inequality gives
\begin{equation}
\begin{aligned}
|x_{123}|^2
&=
\left|
\bra{\psi}\bP_3((23))^{T_3}\bP_3((12))\ket{\psi}
\right|^2\\
&\le
\bra{\psi}\bP_3((23))^{T_3}\ket{\psi}
\bra{\psi}\bP_3((12))\bP_3((23))^{T_3}\bP_3((12))\ket{\psi}\\
&=
\bra{\psi}\bP_3((23))^{T_3}\ket{\psi}
\bra{\psi}\bP_3((13))^{T_3}\ket{\psi}\\
&=
x_{23}x_{13}.
\end{aligned}
\end{equation}
Thus
\begin{equation}\label{eq: bound of Re x_123}
    \operatorname{Re}(x_{123}) \le |x_{123}| \le\sqrt{x_{23}x_{13}}.
\end{equation}

We are now prepared to prove $\tr[W_1\phi] \leq 1$, which is equivalent to show 
\begin{equation}
-2+(d-1)x_{12}+x_{23}+x_{13}+2\operatorname{Re} (x_{123})
\le d+1.
\end{equation}
It remains to prove this for any tripartite stabilizer state $\phi$.
We divide the proof into several cases according to the ordering of
\(m_{12},m_{13},m_{23}\).

\begin{enumerate}
    \item First suppose $m_{12}\ge m_{13},\; m_{12}\ge m_{23}$. 
    The estimates above give
    \begin{equation}
            x_{12}\le1,
    \quad
    x_{23}\le1,
    \quad
    x_{13}\le1,
    \quad
    \operatorname{Re} (x_{123})\le 1.
    \end{equation}
    Therefore
    \begin{equation}
    \begin{aligned}
    -2+(d-1)x_{12}+x_{23}+x_{13}+2\operatorname{Re} (x_{123})
    &\le
    -2+(d-1)+1+1+2\\
    &=
    d+1.
    \end{aligned}
    \end{equation}

    \item Next suppose $m_{23}\ge m_{12}\ge m_{13}$.  
    Then
    \begin{equation}
    x_{12}\le 2^{m_{12}-m_{23}},
\quad
    x_{23}\le 2^{m_{23}-m_{12}},
\quad
    x_{13}\le 2^{m_{13}-m_{23}},\quad \operatorname{Re} (x_{123})
    \le
    2^{(m_{13}-m_{12})/2} \leq 1.
    \end{equation}
    For simplicity, we denote
    \begin{equation}
    r=  2^{m_{23}-m_{12}} \geq 1,
    \quad
    s=2^{m_{13}-m_{12}} \leq 1,
    \quad
    h=2^{m_{\mathrm{GHZ}}+2m_{12}+m_{13}+m_{\mathrm{loc}}}\ge 1,
    \end{equation}
    then we have
    \begin{equation}
       rh =  2^{m_{\mathrm{GHZ}} + m_{12} + m_{13} + m_{23} + m_{\mathrm{loc}}} = d, \quad x_{12} \leq \frac{1}{r}, \quad x_{23} \le r, \quad x_{13} \le \frac{s}{r}.
    \end{equation}

    Therefore
    \begin{equation}
    \begin{aligned}
    d+1 - (-2+(d-1)x_{12}+x_{23}+x_{13}+2\operatorname{Re} (x_{123})) &\ge
    rh+1 - (-2+\frac{rh - 1}{r}+r+\frac{s}{r}+2\sqrt{s})\\
    & \ge 
    (r-1)(h-1)
    +
    2(1-\sqrt{s})
    +
    \frac{1-s}{r} \\
    &\ge 0,
    \end{aligned}
    \end{equation}
    which means
    \begin{equation}    
    -2+(d-1)x_{12}+x_{23}+x_{13}+2\operatorname{Re} (x_{123})
    \le d+1.
    \end{equation}

    \item Suppose $m_{13}\ge m_{12}\ge m_{23}$. This case is symmetric to Case 2 of $m_{23}\ge m_{12}\ge m_{13}$, with \(m_{13}\) and \(m_{23}\) interchanged, and thus \(x_{13}\) and \(x_{23}\) interchanged. Similarly, we set
    \begin{equation}
    r=2^{m_{13}-m_{12}},
    \quad
    s=2^{m_{23}-m_{12}},\quad h=2^{m_{\mathrm{GHZ}}+2m_{12}+m_{23}+m_{\mathrm{loc}}}.
    \end{equation}
    Repeating the same argument as in Case 2 gives
    \[
    -2+(d-1)x_{12}+x_{23}+x_{13}+2\operatorname{Re} (x_{123})
    \le d+1.
    \]

    \item Suppose $m_{23}\ge m_{13}\ge m_{12}$.
    Then
    \begin{equation}
    x_{12}\le 2^{m_{12}-m_{23}}\le 2^{m_{13}-m_{23}},
    \quad
    x_{23}\le 2^{m_{23}-m_{13}},
    \quad
    x_{13}\le 2^{m_{13}-m_{23}},
    \quad
    \operatorname{Re} (x_{123}) \le1.
    \end{equation}
    We denote $r=2^{m_{23}-m_{13}}$,
    then \(1\le r\le d\).
    Thus
    \begin{equation}
    \begin{aligned}
    -2+(d-1)x_{12}+x_{23}+x_{13}+2\operatorname{Re} (x_{123}) &\le
    -2
    +(d-1)2^{m_{13}-m_{23}}
    +2^{m_{23}-m_{13}}
    +2^{m_{13}-m_{23}}
    +2\\
    &=
    \frac{d}{r}+r \\
    & \le d + 1.
    \end{aligned} 
    \end{equation}

    \item Suppose $m_{13}\ge m_{23}\ge m_{12}$. This case is symmetric to Case 4, again by interchanging \(m_{13}\) and \(m_{23}\), and thus \(x_{13}\) and \(x_{23}\). Set $r=2^{m_{13}-m_{23}}.$
    Then \(1\le r\le d\), and the same argument as in Case 4 gives
    \begin{equation}
            -2+(d-1)x_{12}+x_{23}+x_{13}+2\operatorname{Re} (x_{123})
    \le d+1.
    \end{equation}
\end{enumerate}

The above cases exhaust all possible orderings of $m_{12}, m_{13}, m_{23}.$
Therefore, for every pure tripartite stabilizer projector \(\phi\), we conclude that
\begin{equation}
-2+(d-1)x_{12}+x_{23}+x_{13}+2\operatorname{Re} (x_{123})
\le d+1,
\end{equation}
which means $\tr[W_1\phi] \leq 1$. Next we show the lower bound $\tr[W_1\phi] \geq -1$, which is equivalent to show
\begin{equation}
    -2+(d-1)x_{12}+x_{23}+x_{13}+2\operatorname{Re} (x_{123})
\ge -(d+1).
\end{equation}
We know that the spectrum of $\bP_3((12))$ is $\pm 1$, and hence $x_{12} = \tr[\bP_3((12)) \phi] \geq -1$. Moreover, since $x_{23} \geq 0$ and $x_{13} \geq 0$, Eq.~\eqref{eq: bound of Re x_123} gives
\begin{equation}
\begin{aligned}
    x_{23} + x_{13} + 2 \operatorname{Re} (x_{123}) &\geq x_{23} + x_{13} -2 |x_{123}| \\
    &\geq x_{23} + x_{13} -2 \sqrt{x_{23}x_{13}} \\
    &= (\sqrt{x_{23}} - \sqrt{x_{13}})^2\\
    & \ge 0.
\end{aligned}
\end{equation}
Together, we have
\begin{equation}
        -2+(d-1)x_{12}+x_{23}+x_{13}+2\operatorname{Re} (x_{123}) \ge -2-(d-1) + 0
\ge -(d+1).
\end{equation}
This means $\tr[W_1\phi] \geq -1$. Overall, this proves $|\tr[W_1\phi]| \leq 1$ for every pure tripartite stabilizer projector $\phi$.
\end{proof}

The same proof also shows that 
\begin{equation}
    |\tr[W_2\phi]| = |\operatorname{Re} (x_{123})| \le |x_{123}| \le\sqrt{x_{23}x_{13}}.
\end{equation}
Moreover, the estimates above give
\begin{equation}
x_{23}x_{13}
\le
2^{m_{23}-\max\{m_{12},m_{13}\}}
2^{m_{13}-\max\{m_{12},m_{23}\}}
\le 1.
\end{equation}
Thus $|\tr[W_2\phi]| \le 1$.

The above calculation implies that, for any Hermitian operator \(H\) with stabilizer decomposition $H=\sum_j c_j\phi_j$, where \(c_j\in\mathbb R\) and each \(\phi_j\) is a pure stabilizer projector, we have
\begin{equation}
    \begin{aligned}
|\tr[WH]| 
&=
\left|
\sum_j c_j\tr\left[\left( \frac1{d-1}W_1 + \left(1 - \frac1{d-1}\right)W_2 \right) \phi_j \right]
\right| \\
&\le
\frac1{d-1} \left|
\sum_j c_j\tr\left[ W_1 \phi_j \right]
\right| +  \left(1 - \frac1{d-1}\right)  \left|
\sum_j c_j\tr\left[ W_2 \phi_j \right]
\right|  \\
& \le \frac1{d-1}\sum_j |c_j|\,\left|\tr\left[ W_1 \phi_j \right]\right|
+ \left(1 - \frac1{d-1}\right)\sum_j |c_j|\,\left|\tr\left[ W_2 \phi_j \right]\right|
\\
& \le \frac1{d-1}\sum_j |c_j| + \left(1 - \frac1{d-1}\right)\sum_j |c_j|
\\
&= \sum_j |c_j|.
    \end{aligned}
\end{equation}
Taking the minimum over all stabilizer decompositions and applying it to
\(\Psi^\cE_{A_I^2A_O}\), we obtain
\begin{equation}\label{eq:Wsm-bound-robustness}
|\tr[W\Psi^\cE_{A_I^2A_O}]|
\le
\cR(\Psi^\cE_{A_I^2A_O}).
\end{equation}
Thus it remains to lower bound
\(\tr[W\Psi^\cE_{A_I^2A_O}]\).

Recall that $\Psi^\cE_{A_I^2A_O}$ can be represented in the
\(\{S_+,S_-,S_0,S_1\}\) basis as
\begin{equation}
s_i=\tr[\Psi^\cE_{A_I^2A_O}S_i],
\qquad
i\in\{+,-,0,1\},
\end{equation}
where \(s_0,s_1\) are determined by the purification constraints
Eqs.~\eqref{eq:s0-solution} and \eqref{eq:s1-solution}. 
We calculate the coefficients in the \(S\)-basis,
\begin{equation}
\begin{aligned}
    \tr[W_1\Psi^\cE_{A_I^2A_O}] &= \frac1{d+1}((d-3)s_+ - (d+1)s_- + (d-1)s_0 +2ds_1)\\
     \tr[W_2\Psi^\cE_{A_I^2A_O}] &= \frac{1}{2}(s_0 + ds_1).
\end{aligned}
\end{equation}
Thus
\begin{equation}\label{eq:Wsm-S-basis}
\begin{aligned}
\tr[W\Psi^\cE_{A_I^2A_O}]
&= \frac1{d-1} \tr[W_1\Psi^\cE_{A_I^2A_O}] + \left(1 - \frac1{d-1}\right)\tr[W_2\Psi^\cE_{A_I^2A_O}] \\
&= 
\frac{d-3}{(d-1)(d+1)}s_+
-
\frac{1}{d-1}s_- \\
&+
\left(
\frac{1}{d+1}
+
\frac{d-2}{2(d-1)}
\right)s_0
+
\left(
\frac{2d}{(d-1)(d+1)}
+
\frac{d(d-2)}{2(d-1)}
\right)s_1 \\
&= p\left[
1+
\frac{
\left(d^2+d\left(-2(\delta-2)\delta-3\right)+2(\delta-2)\delta+4\right)
\left(d(\delta+f-1)-\delta\right)
}{
2(d-1)(\delta-1)\delta\left(d(\delta-1)-\delta\right)
}
\right] 
+
c_+s_+
+
c_-s_- \\
&= p\left[
1+\frac{
\frac{d-2}{2} + \frac{1}{d-1}+2\delta-\delta^2
}
{\lambda_0\delta(1-\delta)}(f-\lambda_0)
\right]
+
c_+s_+
+
c_-s_- \\
&=
p\left[
1+\cK_\cR^l(f-\lambda_0)
\right]
+
c_+s_+
+
c_-s_-,
\end{aligned}
\end{equation}
where $\cK_\cR^l = \frac{
\frac{d-2}{2} + \frac{1}{d-1}+2\delta-\delta^2
}
{\lambda_0\delta(1-\delta)}$,  \(c_+\) and \(c_-\) are
\begin{equation}
\begin{aligned}
    c_+
&=
\frac{1}
{2\delta(d-1)(d+2)\bigl(d(1-\delta)+\delta\bigr)}\Big( (d-2)^4(\delta^2-2\delta+2) \\
&+ (d-2)^3(2\delta^3-4\delta+10) + (d-2)^2(6\delta^3-9\delta^2+6\delta+20) \\
&+ (d-2)(4\delta^3-8\delta^2+12\delta+24)
+
16 \Big)
\end{aligned}
\end{equation}
and
\begin{equation}
\begin{aligned}
    c_-
&=\frac{(d-1)}
{2(d+1)\bigl(d(1-\delta)+\delta\bigr)} \Big((d-2)^2(2-\delta)
+
(d-2)(-2\delta^2+\delta+6) + 4+8\delta-4\delta^2 \Big)
\end{aligned}
\end{equation}
For \(d=2^m\) and \(0<\delta<1\), the denominators of
\(c_+\) and \(c_-\) are positive.
For the numerator, we check
\begin{equation}
\begin{aligned}
& \delta^2-2\delta+2=(1-\delta)^2+1>0,
\quad
2\delta^3-4\delta+10>0, \\
& 6\delta^3-9\delta^2+6\delta+20>0,
\quad
4\delta^3-8\delta^2+12\delta+24>0.
\end{aligned}
\end{equation}
Since \(d-2\ge0\), this implies the numerator of $c_+$ is positive. 
Similarly,
\begin{equation}
2-\delta>0,
\quad
-2\delta^2+\delta+6>0,
\quad
4+8\delta-4\delta^2>0
\end{equation}
Hence the numerator of $c_-$ is also positive. Therefore $c_+>0,
\;
c_->0.$
Since the positivity criterion in Eq.~\eqref{eq:S-positivity-criterion}
implies $s_+\ge0,
\;
s_-\ge0,$
we conclude that
\begin{equation}
\tr[W\Psi^\cE_{A_I^2A_O}] = p\left[
1+\cK_\cR^l(f-\lambda_0)
\right]
+
c_+s_+
+
c_-s_-
\ge
p\left[
1+\cK_\cR^l(f-\lambda_0)
\right].
\end{equation}
Finally, by Eq.~\eqref{eq:Wsm-bound-robustness},
\begin{equation}
\cR(\Psi^\cE_{A_I^2A_O})
\ge
|\tr[W\Psi^\cE_{A_I^2A_O}]|
\ge
\tr[W\Psi^\cE_{A_I^2A_O}]\ge p\left[
1+\cK_\cR^l(f-\lambda_0)
\right].
\end{equation}
Therefore, every feasible Clifford-twirled Choi state satisfies
\begin{equation}\label{eq:multi-robustness-necessary-sm}
\frac{1}{p}\cR(\Psi^\cE_{A_I^2A_O})
\ge
1+
\cK_\cR^l
(f-\lambda_0).
\end{equation}
For an arbitrary feasible protocol, the Clifford twirl preserves the constraints and does not increase the Choi-state robustness. Hence the same lower bound holds before twirling. Taking the minimum over all feasible protocols gives
\begin{equation}
\cR_{\cD_\delta}(2,f,p,d)
\ge
1+
\cK_\cR^l
(f-\lambda_0).
\end{equation}
Combining this lower bound with the upper bound, we conclude that
\begin{equation}
1+
\cK_\cR^l
(f-\lambda_0)
\le \cR_{\cD_\delta}(2,f,p,d)  \le 
1+
\cK_\cR^u
(f-\lambda_0).
\end{equation}
where
\begin{equation}
    \cK_\cR^l = \frac{
\frac{d-2}{2}
+
\frac{1}{d-1} + 2\delta-\delta^2
} {\lambda_0\delta(1-\delta)} , \quad \cK_\cR^u = \frac{2d-3 + 2\delta-\delta^2}
{\lambda_0\delta(1-\delta)}.
\end{equation}

In particular, for a single-qubit system $m=1$, note that $2d-3 = 1, \frac{d-2}{2}
+
\frac{1}{d-1} = 1$. Thus the upper bound and lower bound collapse, i.e.
\begin{equation}
    \cR_{\cD_\delta}(2,f,p,2) = 1+
\frac{1 + 2\delta-\delta^2}
{\lambda_0\delta(1-\delta)}
(f-\lambda_0).
\end{equation}
This completes the proof. 
\end{proof}


\end{document}

%% file: pretex.tex
\usepackage[dvipsnames]{xcolor}
\usepackage{framed}
\definecolor{shadecolor}{rgb}{0.9,0.9,0.9}

\usepackage{mathtools}
\usepackage{amsmath}
\usepackage[shortlabels]{enumitem}

\usepackage{graphicx,epic,eepic,epsfig,amsmath,latexsym,amssymb,verbatim,color}
 
\usepackage{amsfonts}       
\usepackage{nicefrac}       

\usepackage{amsmath}
\usepackage{bbm}

\usepackage{float}
\usepackage{tikz}
\usetikzlibrary{chains}
\usetikzlibrary{fit}
\usepackage{pgflibraryarrows}		
\usepackage{pgflibrarysnakes}		

\usepackage{epsfig}
\usetikzlibrary{shapes.symbols,patterns} 
\usepackage{pgfplots}

\usepackage[strict]{changepage}
\usepackage{hyperref}
\hypersetup{colorlinks=true,citecolor=blue,linkcolor=blue,filecolor=blue,urlcolor=blue,breaklinks=true}

\usepackage[marginal]{footmisc}
\usepackage{url}
\usepackage{theorem}

\newtheorem{definition}{Definition}
\newtheorem{proposition}{Proposition}
\newtheorem{lemma}[proposition]{Lemma}

\newtheorem{theorem}[proposition]{Theorem}

\newtheorem{corollary}[proposition]{Corollary}

\def\squareforqed{\hbox{\rlap{$\sqcap$}$\sqcup$}}
\def\qed{\ifmmode\squareforqed\else{\unskip\nobreak\hfil
\penalty50\hskip1em\null\nobreak\hfil\squareforqed
\parfillskip=0pt\finalhyphendemerits=0\endgraf}\fi}
\def\endenv{\ifmmode\;\else{\unskip\nobreak\hfil
\penalty50\hskip1em\null\nobreak\hfil\;
\parfillskip=0pt\finalhyphendemerits=0\endgraf}\fi}
\newenvironment{proof}{\noindent \textbf{{Proof~} }}{\hfill $\blacksquare$}

\newcounter{remark}

\newcounter{example}

\mathchardef\ordinarycolon\mathcode`\:
\mathcode`\:=\string"8000
\def\vcentcolon{\mathrel{\mathop\ordinarycolon}}
\begingroup \catcode`\:=\active
  \lowercase{\endgroup
  \let :\vcentcolon
  }

\usepackage{cleveref}
\usepackage{graphicx}
\usepackage{xcolor}

\RequirePackage[framemethod=default]{mdframed}
\newmdenv[skipabove=7pt,
skipbelow=7pt,
backgroundcolor=darkblue!15,
innerleftmargin=5pt,
innerrightmargin=5pt,
innertopmargin=5pt,
leftmargin=0cm,
rightmargin=0cm,
innerbottommargin=5pt,
linewidth=1pt]{tBox}

\newmdenv[skipabove=7pt,
skipbelow=7pt,
backgroundcolor=red!15,
innerleftmargin=5pt,
innerrightmargin=5pt,
innertopmargin=5pt,
leftmargin=0cm,
rightmargin=0cm,
innerbottommargin=5pt,
linewidth=1pt]{rBox}

\newmdenv[skipabove=7pt,
skipbelow=7pt,
backgroundcolor=blue2!25,
innerleftmargin=5pt,
innerrightmargin=5pt,
innertopmargin=5pt,
leftmargin=0cm,
rightmargin=0cm,
innerbottommargin=5pt,
linewidth=1pt]{dBox}
\newmdenv[skipabove=7pt,
skipbelow=7pt,
backgroundcolor=darkkblue!15,
innerleftmargin=5pt,
innerrightmargin=5pt,
innertopmargin=5pt,
leftmargin=0cm,
rightmargin=0cm,
innerbottommargin=5pt,
linewidth=1pt]{sBox}
\definecolor{darkblue}{RGB}{0,76,156}
\definecolor{darkkblue}{RGB}{0,0,153}
\definecolor{blue2}{RGB}{102,178,255}
\definecolor{darkred}{RGB}{195,0,0}

\newcommand{\nc}{\newcommand}
\nc{\rnc}{\renewcommand}
\nc{\lbar}[1]{\overline{#1}}
\nc{\bra}[1]{\langle#1|}
\nc{\ket}[1]{|#1\rangle}
\nc{\ketbra}[2]{|#1\rangle\!\langle#2|}
\nc{\braket}[2]{\langle#1|#2\rangle}

\nc{\proj}[1]{| #1\rangle\!\langle #1 |}
\nc{\avg}[1]{\langle#1\rangle}
\nc{\smfrac}[2]{\mbox{$\frac{#1}{#2}$}}
\nc{\tr}{\operatorname{Tr}}
\nc{\ox}{\otimes}
\nc{\dg}{\dagger}
\nc{\dn}{\downarrow}
\nc{\cA}{{\cal A}}
\nc{\cB}{{\cal B}}
\nc{\cC}{{\cal C}}
\nc{\cD}{{\cal D}}
\nc{\cE}{{\cal E}}
\nc{\cF}{{\cal F}}
\nc{\cG}{{\cal G}}
\nc{\cH}{{\cal H}}
\nc{\cI}{{\cal I}}
\nc{\cJ}{{\cal J}}
\nc{\cK}{{\cal K}}
\nc{\cL}{{\cal L}}
\nc{\cM}{{\cal M}}
\nc{\cN}{{\cal N}}
\nc{\cO}{{\cal O}}
\nc{\cP}{{\cal P}}
\nc{\cQ}{{\cal Q}}
\nc{\cR}{{\cal R}}
\nc{\cS}{{\cal S}}
\nc{\cT}{{\cal T}}
\nc{\cU}{{\cal U}}
\nc{\cV}{{\cal V}}
\nc{\cX}{{\cal X}}
\nc{\cY}{{\cal Y}}
\nc{\cZ}{{\cal Z}}
\nc{\cW}{{\cal W}}
\nc{\csupp}{{\operatorname{csupp}}}
\nc{\qsupp}{{\operatorname{qsupp}}}
\nc{\var}{{\operatorname{var}}}
\nc{\rar}{\rightarrow}
\nc{\lrar}{\longrightarrow}
\nc{\polylog}{{\operatorname{polylog}}}
\nc{\wt}{{\operatorname{wt}}}
\nc{\av}[1]{{\left\langle {#1} \right\rangle}}
\nc{\supp}{{\operatorname{supp}}}

\nc{\argmin}{{\operatorname{argmin}}}

\def\x{\xi}

\nc{\RR}{{{\mathbb R}}}
\nc{\CC}{{{\mathbb C}}}
\nc{\FF}{{{\mathbb F}}}
\nc{\NN}{{{\mathbb N}}}
\nc{\ZZ}{{{\mathbb Z}}}
\nc{\PP}{{{\mathbb P}}}
\nc{\QQ}{{{\mathbb Q}}}
\nc{\UU}{{{\mathbb U}}}
\nc{\EE}{{{\mathbb E}}}
\nc{\id}{{\operatorname{id}}}

\nc{\CHSH}{{\operatorname{CHSH}}}
\newcommand{\bP}{\mathbf{P}}

\nc{\be}{\begin{equation}}
\nc{\ee}{{\end{equation}}}
\nc{\bea}{\begin{eqnarray}}
\nc{\eea}{\end{eqnarray}}
\nc{\<}{\langle}
\rnc{\>}{\rangle}
\nc{\rU}{\mbox{U}}

\nc{\ob}[1]{#1}

\nc{\SEP}{{\text{\rm SEP}}}
\nc{\NS}{{\text{\rm NS}}}
\nc{\LOCC}{{\text{\rm LOCC}}}
\nc{\PPT}{{\text{\rm PPT}}}
\nc{\EXT}{{\text{\rm EXT}}}
\nc{\Sym}{{\operatorname{Sym}}}

\nc{\ERLO}{{E_{\text{r,LO}}}}
\nc{\ERLOCC}{{E_{\text{r,LOCC}}}}
\nc{\ERPPT}{{E_{\text{r,PPT}}}}
\nc{\ERLOCCinfty}{{E^{\infty}_{\text{r,LOCC}}}}
\nc{\Aram}{{\operatorname{\sf A}}}

\usepackage{tikz}
\usepackage{hyperref}
\hypersetup{colorlinks=true,citecolor=blue,linkcolor=blue,filecolor=blue,urlcolor=blue,breaklinks=true}

\makeatletter
\def\grd@save@target#1{%
  \def\grd@target{#1}}
\def\grd@save@start#1{%
  \def\grd@start{#1}}
\tikzset{
  grid with coordinates/.style={
    to path={%
      \pgfextra{%
        \edef\grd@@target{(\tikztotarget)}%
        \tikz@scan@one@point\grd@save@target\grd@@target\relax
        \edef\grd@@start{(\tikztostart)}%
        \tikz@scan@one@point\grd@save@start\grd@@start\relax
        \draw[minor help lines,magenta] (\tikztostart) grid (\tikztotarget);
        \draw[major help lines] (\tikztostart) grid (\tikztotarget);
        \grd@start
        \pgfmathsetmacro{\grd@xa}{\the\pgf@x/1cm}
        \pgfmathsetmacro{\grd@ya}{\the\pgf@y/1cm}
        \grd@target
        \pgfmathsetmacro{\grd@xb}{\the\pgf@x/1cm}
        \pgfmathsetmacro{\grd@yb}{\the\pgf@y/1cm}
        \pgfmathsetmacro{\grd@xc}{\grd@xa + \pgfkeysvalueof{/tikz/grid with coordinates/major step}}
        \pgfmathsetmacro{\grd@yc}{\grd@ya + \pgfkeysvalueof{/tikz/grid with coordinates/major step}}
        \foreach \x in {\grd@xa,\grd@xc,...,\grd@xb}
        \node[anchor=north] at (\x,\grd@ya) {\pgfmathprintnumber{\x}};
        \foreach \y in {\grd@ya,\grd@yc,...,\grd@yb}
        \node[anchor=east] at (\grd@xa,\y) {\pgfmathprintnumber{\y}};
      }
    }
  },
  minor help lines/.style={
    help lines,
    step=\pgfkeysvalueof{/tikz/grid with coordinates/minor step}
  },
  major help lines/.style={
    help lines,
    line width=\pgfkeysvalueof{/tikz/grid with coordinates/major line width},
    step=\pgfkeysvalueof{/tikz/grid with coordinates/major step}
  },
  grid with coordinates/.cd,
  minor step/.initial=.2,
  major step/.initial=1,
  major line width/.initial=2pt,
}
\makeatother

\usepackage{thmtools}
\usepackage{thm-restate}
\usepackage{etoolbox}
\makeatletter
\def\problem@s{}
\newcounter{problems@cnt}

\newcommand{\allproblems}{\problem@s}
\makeatother